\documentclass[pra, twocolumn, a4paper, aps, 10pt, nofootinbib]{revtex4-2}
\usepackage[parfill]{parskip}
\usepackage[margin=50pt]{geometry}
\usepackage[utf8]{inputenc}
\usepackage{physics}
\usepackage{tikz}
\usetikzlibrary{quantikz2}
\usepackage{enumitem}
\usepackage{caption}
\usepackage{subcaption}
\usepackage{amsmath}
\usepackage{amsthm}
\usepackage{amssymb} 
\newtheorem{definition}{Definition}[section]
\newtheorem{theorem}{Theorem}[section]
\newtheorem{conjecture}{Conjecture}[section]
\newtheorem{corollary}{Corollary}[section]
\newtheorem{lemma}{Lemma}[section]
\theoremstyle{definition}

\usepackage{comment}
\usepackage[ruled]{algorithm2e}
\SetAlgorithmName{Protocol}{protocol}{List of Protocols}
\SetAlgoCaptionSeparator{ ---}
\usepackage{graphicx}
\usepackage{color}
\definecolor{ulbblue}{HTML}{003087} 
\usepackage{hyperref}
\usepackage{dutchcal}
\usepackage{mathtools}
\usepackage{thmtools}
\usepackage{bbm}

\hypersetup{
  colorlinks = true, 
  urlcolor  = blue!50!black, 
  linkcolor = ulbblue!70, 
  citecolor = blue!30!red 
}

\newcommand\id{\mathbbm{1}}
\DeclareMathOperator*{\Motimes}{\text{\raisebox{0.25ex}{\scalebox{0.8}{$\bigotimes$}}}}

\def\ket#1{\mathinner{|{#1}\rangle}}
\def\braket#1#2{\mathinner{\langle{#1|#2}\rangle}}
\def\ketbra#1#2{\mathinner{|{#1}\rangle\langle{#2}|}}
\def\dketbra#1#2{\mathinner{|{#1}\rangle\!\rangle\!\langle\!\langle{#2}|}}
\def\dket#1{\mathinner{|{#1}\rangle\!\rangle}}

\begin{document}
\hypersetup{pdfborder=0 0 0}
\title{Characterizing unitaries via quasi-process functions}
\author{Nasra \surname{Daher Ahmed}}
\email{ndaherah@proton.me}
\author{Ravi Kunjwal}
\email{quaintum.research@gmail.com}
\affiliation{Aix Marseille University, CNRS, LIS, Marseille, France}

\begin{abstract}
\noindent
It is well-known that every $n$-qubit unitary can be decomposed into a
product of controlled single-qubit unitaries. We show that the control
conditions in such a decomposition naturally define a quasi-process function, an object that originally
came up in the study of indefinite causal order. Here quasi-process functions yield a
representation of every $n$-qubit unitary in a canonical form, namely,
as a quasi-process function unitary (qPFU). We use this representation to show how quantum nonlocality without entanglement (QNLWE) can be understood in terms of unitaries that encode the trade-off between causal order and locality noted in R.~Kunjwal and \"A.~Baumeler, \href{https://doi.org/10.1103/PhysRevLett.131.120201}{Phys.~Rev.~Lett.~131, 120201 (2023)}. Specifically, we obtain two sufficient (but not necessary) criteria---mutual exclusivity of control conditions (MECC) and Unambiguity---for a unitary to be a process function unitary (PFU), \textit{i.e.}, it admits a qPFU representation where the quasi-process function is a process function, an object that can model paradox-free classical causal loops. We also show the inequivalence of these criteria, neither implying the other. We then show that the circuit depth of MECC unitaries admits a characterization in terms of clique-partitioning problems on Hamming graphs and use this correspondence to provide bounds on the circuit depth of $n$-qubit MECC unitaries.
Finally, we show how MECC unitaries provide a recipe
for constructing a unitary purification of the associated process
function, thereby providing a large class of unitary processes of interest in quantum causality. The generality of this causality-inspired representation of unitaries makes it a versatile tool for analyzing the structure of quantum circuits in other domains.
\end{abstract}	

\maketitle
\tableofcontents

\section{Introduction}
Any quantum experiment can be understood as a composition of three basic procedures: preparations, transformations, and measurements. In the circuit model of quantum computation, the preparation is typically an initialization of $n$ qubits in the $\ket{0}^{\otimes n}$ state, the transformation is a unitary $U$ acting on the $n$ qubits, and the measurement is typically a computational basis measurement, $Z^{\otimes n}$. Our concern here is the structure of the unitary that sits at the heart of the circuit model. It is this unitary that allows the qubits to ``talk" to each other in pursuit of the intended goal of the computation. It is this ``talking" that we model here via quasi-process functions \cite{KO23}, noting how all $n$-qubit unitaries are \textit{quasi-process function unitaries} (qPFUs).

To motivate the qPFU representation, we need to first take a detour into quantum nonlocality without entanglement (QNLWE) \cite{BDF99} and classical instances of indefinite causal order \cite{BW16}, two notions that were previously shown to be intimately linked \cite{KB23,DSL25}. This detour will motivate how the qPFU representation naturally arises from a particular way of thinking about this connection, even as the  representation itself is defined independently of it. Indeed, the concrete use-case of this representation that we will focus on in this paper is the trade-off between causal order and locality represented by QNLWE \cite{KB23}.

QNLWE refers to any ensemble of multipartite mutually orthogonal product states which cannot be perfectly distinguished from one another when the parties are limited to local operations and classical communication (LOCC), implying a gap between the set of separable operations (SEP) and LOCC. One example of QNLWE found in the seminal work of Bennett \textit{et al.}~\cite{BDF99} is the SHIFT basis:
\begin{align}
\begin{split}
&\{\ket{000}, \ket{111}, \ket{+01}, \ket{-01}, \\  
&\ket{1+0},\ket{1-0},\ket{01+}, \ket{01-} \},      
\end{split} 
\end{align}

where $\ket{\pm} = \frac{1}{\sqrt{2}}(\ket{0} \pm \ket{1})$.

While it is possible to prepare these states locally, it is not possible to perfectly discriminate between them using local operations and classical communication. A closer examination reveals that this conflict can be traced to an implicit assumption in LOCC, namely, that of definite causal order between the parties. This is reflected in the fact that at least one party must start any  LOCC protocol, but any information gained by this party is always at the expense of perfect distinguishability \cite{GV01}. It is the implicit assumption of definite causal order in LOCC that was exploited in Ref.~\cite{KB23} to reveal how perfect discrimination of the states in the SHIFT basis is possible using only local operations if we allow the classical communication to be carried out in an indefinite causal order.

The classical communication  exploited in Ref.~\cite{KB23} is given by a deterministic classical process (also known as a \textit{process function} \cite{BT21}) named the Araújo-Feix/Baumeler-Wolf (AF/BW) process \cite{BW16}:
\begin{equation}\label{lugano}
i_1 = (o_2 \oplus 1)o_3, i_2 = (o_3 \oplus 1)o_1, i_3 = (o_1 \oplus 1)o_2,
\end{equation}
where party $k$ receives a binary input $i_k\in\{0,1\}$ from the process and then sends a binary output $o_k\in\{0,1\}$ to the process ($k\in\{1,2,3\}$) (see Figure \ref{fig:afbw}). Since the input of each party is a function of the outputs of the other two parties, every party is in the causal future of the other two, \textit{i.e.}, no global causal order can be defined between them. 

\begin{figure}[htb!]
\centering
\includegraphics[scale=0.6]{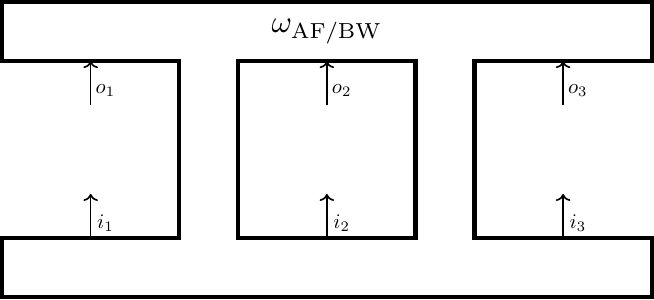}
\caption{AF/BW process function \cite{BW16}: Each party $k\in\{1,2,3\}$ receives an input $i_k\in\{0,1\}$ and sends an output $o_k\in\{0,1\}$. The process function wires the outputs to the inputs, allowing the parties to communicate according to Eq.~\eqref{lugano}.}
\label{fig:afbw}
\end{figure}

The trade-off between causal order and locality extends well beyond the case of the SHIFT ensemble and its multiqubit generalization was already proved in Ref.~\cite{KB23}. More recently, a characterization of QNLWE ensembles that admit perfect discriminability via local operations and process functions (LOPF) was obtained in Ref.~\cite{DSL25}. This characterization refers to a property of QNLWE ensembles termed  \textit{unambiguity}, \textit{i.e.}, the possibility of partitioning each party's local vectors in a QNLWE ensemble into disjoint subsets of mutually orthogonal vectors \cite{DSL25}.

While the connection between process functions and QNLWE is thus well-understood \cite{KB23, DSL25}, there is another aspect that merits careful scrutiny: the structure of unitaries that rotate the computational basis to the QNLWE bases of interest. The structure of these unitaries can encode process functions and we will refer to such unitaries as \textit{process function unitaries} (PFUs). These unitaries provide the obvious (causal) protocol for perfect discrimination of QNLWE bases that gives up on the locality of operations in LOCC: namely, by subjecting the (composite) quantum system to the adjoint of the unitary and subsequently measuring it in the computational basis. This is distinct from the (noncausal) protocol that perfectly discriminates between elements of an unambiguous QNLWE basis using only local operations and process functions (LOPF) \cite{KB23,DSL25}. We illustrate the two protocols for the case of the SHIFT basis in Figure \ref{fig:SHIFT protocols}.
\begin{figure}[htb!]
    \begin{subfigure}[t]{0.2\textwidth}
        \centering
        \includegraphics[scale=0.75]{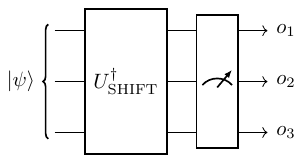}   
        \caption{Causal and global protocol.}
        \label{fig:SHIFT-PFU}
     \end{subfigure}%
     ~
     \begin{subfigure}[t]{0.3\textwidth}
        \centering
        \includegraphics[scale=0.55]{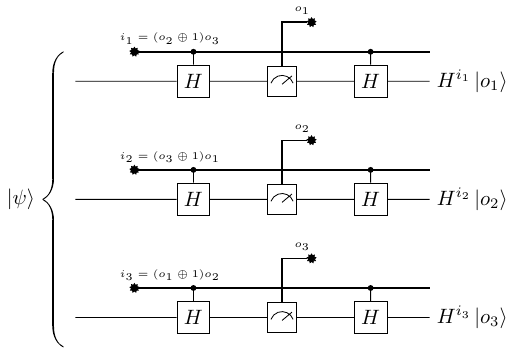}    
        \caption{Noncausal and local \newline protocol \cite{KB23}.}

        \label{fig:SHIFT-LUPF}
     \end{subfigure}
    \caption{Two protocols for the perfect discrimination of the states in the SHIFT basis. $U^{\dagger}_{\rm SHIFT}$ is the unitary that rotates the SHIFT basis to the computational basis.}
    \label{fig:SHIFT protocols}
\end{figure}

Key to our investigation is the conceptual and mathematical connection between these two kinds of protocols: PFUs encode the trade-off between causal order and locality that an LOPF protocol for perfect state discrimination would require. We are interested in finding the conditions under which a unitary is a PFU and, more generally, in the relationships between QNLWE, process functions, and unitaries (Figure \ref{fig:triangle}). Our starting point is the qPFU representation of unitaries.

\begin{figure}[htb!]
\centering
\includegraphics[scale = 0.8]{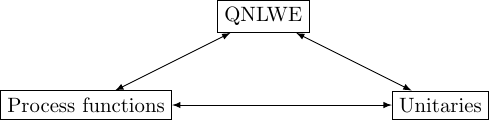}
\caption{\label{fig:triangle}}
\end{figure}

The rest of the paper is organized as follows. In Section \ref{sec:prelim}, we recall some preliminary notions that are needed in the rest of the paper. In Section \ref{sec:qpfu}, we define the qPFU representation and prove that it is general enough to represent all $n$-qubit unitaries. In Section \ref{sec:mecc}, we define unitaries that satisfy the criterion of mutual exclusivity of control conditions (MECC) and show that all MECC unitaries are process function unitaries. We also show that this criterion, while certifying a large class of previously known constructions \cite{NC06} as PFUs, is not necessary via an explicit counter-example. In Section \ref{sec:depth}, we show how the depth of MECC unitaries (which are PFUs) is related to clique-partitioning problems on Hamming graphs. In Section \ref{sec:meccunamb}, we investigate the connection between MECC and the unambiguity criterion of Ref.~\cite{DSL25}. We show that the set of  unitaries associated with unambiguous QNLWE bases (we term them ``Unambiguous Unitaries" or ``Un-U") is distinct from the set of MECC unitaries (MECC-U), neither set contains the other, they have a non-empty intersection, and both sets are strictly contained inside the set of PFUs (see Figure \ref{fig:venn}). In this sense, MECC does not imply unambiguity, unambiguity does not imply MECC, and, while both of them are sufficient for a unitary to be a PFU, neither of them is necessary. In Section \ref{sec:uniprocess}, we show how a MECC unitary can be converted into a unitary process in the sense of the process-matrix framework \cite{OCB12,AFN17} by adapting and generalizing the protocol of Ref.~\cite{KB23}.
We conclude with a discussion of our results and opportunities for future work in Section \ref{sec:conclusion}.  
\begin{figure}[htb!]
\centering
\includegraphics[scale=1]{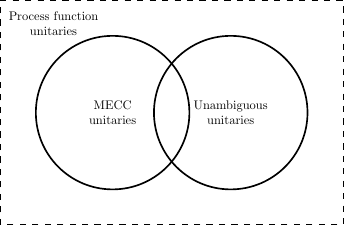}
\caption{Relationship between MECC, unambiguous, and process function unitaries.}
\label{fig:venn}
\end{figure}

\section{Preliminaries} \label{sec:prelim}
We recall below the concepts that are necessary to establish our key results. 
\subsection{Controlled unitary decomposition}
Central to our qPFU representation of unitaries is the notion of controlled unitaries. It is known that any $n$-qubit unitaries can be decomposed in terms of such unitaries:
\begin{theorem}[Section 4.5.2 of \cite{NC10} and \cite{BBC95}]\label{th:controlled-decomposition}
Any $n$-qubit unitary can be decomposed as a product of $(n-1)$-controlled one-qubit unitaries.
\end{theorem}
\subsection{Classical (quasi-)processes}
Let us consider the following scenario: we have $n$ parties, labeled by $k \in [n] := \{1, \dots, n\}$, such that party $k$ receives from the environment a classical input system described by a random variable $I_k$ (taking values $i_k \in \{0, \dots, |I_k| - 1\}$) and sends out to the environment a classical output system described by a random variable $O_k$ (taking values $o_k \in \{0, \dots, |O_k| - 1\})$. The environment is described by the conditional probability distribution $P_{\vec{I}|\vec{O}} : \vec{I} \times \vec{O} \rightarrow [0, 1]$, where $\vec{I} = I_1 \times\dots\times I_n$ and $O = O_1 \times\dots\times O_n$. The parties can interact with the environment through local operations $h = (h_k : I_k \rightarrow O_k)_{k \in [n]}$. 

\begin{figure}[htb!]
\centering
\includegraphics[scale=0.6]{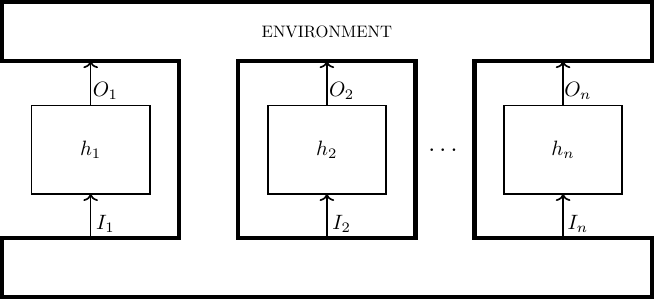}
\caption{General classical process composed with local interventions, \textit{e.g.}, as in Eq.~\eqref{eq:unique-fp}.}
\label{fig:process}
\end{figure}

We now define process functions using the same nomenclature as in Ref.~\cite{KO23}, based on a characterization by \"A.~Baumeler and S.~Wolf's in terms of fixed points \cite{BW16fp}. When the environment is an arbitrary conditional probability distribution, we call it a \textit{classical quasi-process}. If the environment is deterministic, it is described by a \textit{quasi-process function} $\omega = (\omega_k : \vec{O} \rightarrow I_k)_{k \in [n]}$ such that $P_{\vec{I}|\vec{O}} = \prod\limits_{k = 1}^n \delta_{i_k,\omega_k(\vec{o})}$. 
Otherwise, it is described by a probabilistic mixture of quasi-process functions, \textit{i.e.}, $P_{\vec{I}|\vec{O}} = \sum\limits_\lambda p_\lambda \prod\limits_{k = 1}^n \delta_{i_k,\omega^\lambda_k(\vec{o})}$. A classical quasi-process which is logically consistent is a \textit{classical process} \cite{BW16, BW16fp}. Logical consistency refers to the fact that, for some convex decomposition of the classical quasi-process, the map obtained from composing it with any choice of local interventions $h$ admits on average one fixed-point:
\begin{equation}\label{eq:logical-consistency}
\forall h, \sum\limits_\lambda p_\lambda n_\lambda = 1,
\end{equation}
where $n_\lambda = |\{\vec{i} | \vec{i} = \omega^\lambda \circ h(\vec{i}) \}|$ is the number of fixed-points of the composition $\omega^\lambda \circ h$. The quasi-process function  associated to a deterministic classical process is a \textit{process function}. In this case, logical consistency is equivalent to a \textit{unique} fixed-point condition:
\begin{theorem}[Unique fixed-point condition \cite{BW16fp}] \label{fixed-point}
A function $\omega : \vec{O} \rightarrow \vec{I}$ is a valid n-party process function if and only if for every choice of local interventions $h = (h_k : I_k \rightarrow O_k)_{k \in [n]}$, there is a unique fixed point when composing $\omega$ with h, i.e.,
\begin{equation}\label{eq:unique-fp}
    \forall h, \exists! \vec{i} \in \vec{I}: \vec{i} = \omega\circ h(\vec{i}).
\end{equation}
\end{theorem}
This means that, for every possible set of local interventions $h$, there is exactly one global input $\vec{i}$ that is returned by the process function.

\subsection{Quantum nonlocality without entanglement via process functions}
Assume that multiple distant parties share a state picked from an ensemble of pure orthogonal product states $E = \{\ket{\psi_i}\}_i$. Quantum nonlocality without entanglement (QNLWE) refers to the impossibility of perfect state discrimination---that is, identifying which state in the ensemble $E$ the parties share---when the parties are only allowed to perform local operations and communicate classically (LOCC) \cite{BDF99}. Over time, this property has been extensively studied, from its use in quantum data hiding \cite{DPL02} to its extension to a strong notion of QNLWE \cite{HBA19}, and in the case of entangled and/or mixed states \cite{B10, B11,HR25}.

While some QNLWE ensembles like the SHIFT states form a complete basis, others do not. One example of the latter type is given by unextendible product bases (UPBs, see Ref.~\cite{BDP99UPB, DPM03UPB}), which are ensembles of product states that cannot be extended to a product basis. For example, one can construct a UPB from the SHIFT basis by adding the state $\ket{+++}$ to the SHIFT basis states $\{\ket{-01}, \ket{1-0}, \ket{01-}\}$\cite{NC06}. Our focus in this paper is on QNLWE ensembles that do form a complete basis.

Certain QNLWE bases like the SHIFT basis admit a perfect discrimination protocol under local operations and process functions (LOPFs) (see Figure \ref{fig:SHIFT-LUPF} and Ref.~\cite{KB23} for details). Equivalently, the perfect discrimination can be achieved using the adjoint of the unitary $U_\text{SHIFT}$ transforming the computational basis into the SHIFT basis (see Figures \ref{fig:SHIFT-PFU} and \ref{fig:SHIFT-circuit}). This unitary can be said to encode a process function in the sense that the way it rotates a given input state from the computational basis is described by that process function. A rigorous definition of this encoding is provided in Section \ref{sec:qpfu} (Definition \ref{def:qPFU}).

\subsection{Niset-Cerf QNLWE construction}
Of special interest to us in this paper will be a construction of QNLWE ensembles due to Niset and Cerf \cite{NC06}. The Niset-Cerf construction applies to the case of $n$ parties, each holding a system of dimension at least $n-1$. In the $n=3$ case, it leads to the SHIFT ensemble. For $n>3$, however, it gives new examples of multiqudit QNLWE ensembles that are not captured by the multiqubit ensembles constructed in Ref.~\cite{KB23}. 

Recall that the tripartite SHIFT ensemble is given by
\begin{equation}
B_S := \{\ket{000}, \ket{111}, \ket{\pm01}, \ket{1\pm0},\ket{01\pm}\}.
\end{equation}
Expressed in this form, it becomes apparent that $B_S$ can be obtained from the computational basis if each party applies a Hadamard gate controlled by the state of the remaining two parties. The corresponding unitary circuit is represented in Figure \ref{fig:SHIFT-circuit}, where we explicitly label the control conditions. 

\begin{figure}[htb!]
    \centering
    \begin{quantikz}
\lstick{$A$} & \gate{H} & \ctrl{2}  \wire[l][1]["\ket{1}"{above, pos=0.2}]{a} &  \octrl{2}  \wire[l][1]["\ket{0}"{above, pos=0.2}]{a} &  \\
\lstick{$B$}& \octrl{1}  \wire[l][1]["\ket{0}"{above, pos=0.2}]{a}\wire[u]{q} & \gate{H} & \control{} \wire[l][1]["\ket{1}"{above, pos=0.2}]{a} &  \\
\lstick{$C$} & \control{} \wire[l][1]["\ket{1}"{above, pos=0.2}]{a} & \ocontrol{} \wire[l][1]["\ket{0}"{above, pos=0.2}]{a} & \gate{H}  &
\end{quantikz}
    \caption{\label{fig:SHIFT-circuit} SHIFT unitary circuit.}
\end{figure}
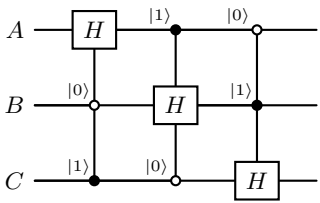

Starting from this example, Niset-Cerf constructed a general family of QNLWE ensembles for $n \geq 3$. Here, each party $k\in\{1,2,\dots,n\}$ has a corresponding computational basis $B_k = \{\ket{0}, \dots, \ket{d_k-1}\}$, with local dimension $d_k \geq n -1$. The Niset-Cerf QNLWE ensemble is then obtained when each party $k$ applies the generalization of the Hadamard gate---the $d_k$-dimensional quantum Discrete Fourier Transform $H_{d_k}$---controlled on the state of the remaining $n-1$ parties. In the case of three parties each holding a qubit, their construction corresponds exactly to the SHIFT unitary in Figure \ref{fig:SHIFT-circuit}. It consists in applying a global unitary $U = U_3U_2U_1$ to the joint computational basis such that 
\begin{align}\label{eq:niset-cerf}
U_1 &:= H \otimes \ketbra{01}{01} + \id_1 \otimes \big(\id_{23} - \ketbra{01}{01}\big)\\
   U_2 &:= \ketbra{1}{1}\otimes H\otimes\ketbra{0}{0}\\ \nonumber
   &+ (\id_{13} - \ketbra{10}{10})\otimes \id_2  \\
   U_3 &:= \ketbra{01}{01}\otimes H + (\id_{12} -\ketbra{01}{01})\otimes \id_3,
\end{align}
where $\id_k$ is the identity operator on the $k$-th system, and $\id_{ij}$ ($i,j\neq k$) is the identity operator on the remaining systems. 

It is worth mentioning that the ordering of the terms in the global unitary is irrelevant since the controlled unitaries commute. This comes from the fact that the controlled unitaries are \textit{mutually exclusive}, a property that we define in Section \ref{sec:mecc}. The decomposition into these three unitaries makes explicit how the SHIFT unitary encodes the AF/BW process. We will see that it corresponds exactly to its qPFU decomposition (as in Definition \ref{def:qPFU}), with the quasi-process function being the AF/BW process.

\section{Quasi-process function unitaries}\label{sec:qpfu}
In this section, we demonstrate that any \textit{n}-qubit unitary can be decomposed in a canonical way into a product of controlled single-qubit unitaries such that the way the controlled unitaries are activated is described by a quasi-process function. Such unitaries belong to the class of \textit{quasi-process function unitaries}, formally defined below:
\begin{definition}[Quasi-process function unitary (qPFU)]\label{def:qPFU}
	An $n$-partite quasi-process function unitary is a unitary $U$ that admits a decomposition of the following form:
	\begin{align}\label{eq:qpfu}
			U=U^{(1)}U^{(2)}\dots U^{(n)},
	\end{align}
where each $U^{(k)}$ ($k\in\{1,2,\dots,n\}$) is a product of controlled unitaries acting on the $k$-th system, \textit{i.e.}, 
\begin{align}
	U^{(k)}&=\prod_{i_k: \exists \vec{o}_{\backslash k}, i_k=\omega_k(\vec{o}_{\backslash k})}U^{(k,i_k)},
\end{align}
with each controlled unitary being defined as
\begin{align}\label{eq:local-U_PFU}
	U^{(k,i_k)}&=\left(\sum_{\vec{o}_{\backslash k}:\omega_k(\vec{o}_{\backslash k})=i_k}\ketbra{\vec{o}_{\backslash k}}{\vec{o}_{\backslash k}}\right)\otimes U^{(k)}_{i_k}\nonumber\\
	&+\left(\id_{\backslash k}-\sum_{\vec{o}_{\backslash k}:\omega_k(\vec{o}_{\backslash k})=i_k}\ketbra{\vec{o}_{\backslash k}}{\vec{o}_{\backslash k}}\right)\otimes 
    U_k,
\end{align}
where $\omega:=(\omega_k:\vec{O} \rightarrow I_k)_{k=1}^n$ defines a quasi-process function from $\vec{O}$ to $\vec{I}$. 

Here $U^{(k)}_{i_k}$ is a unitary that acts on the $k$-th target system if and only if the other systems are in a control state $\ket{\vec{o}_{\backslash k}}$ such that $i_k=\omega_k(\vec{o}_{\backslash k})$, $U_k$ is a fixed unitary acting on the $k$-th system (corresponding to $i_k = 0$), and $\id_{\backslash k}$ is the identity operator on the remaining systems.  
\end{definition}

\begin{theorem}\label{th:qPFU-decomposition}
Any unitary $U$ acting on \textit{n} qubits can be written as a quasi-process function unitary.
\end{theorem}
\begin{proof}
From Theorem \ref{th:controlled-decomposition}, we know that any $n$-qubit unitary can be decomposed as a product of ($n$-1)-controlled single-qubit unitaries. Moreover, any two controlled unitaries acting on the same qubit can be merged if they are activated by the same control state. Hence, for each qubit $k \in [n]$, we label the distinct nontrivial controlled unitaries acting on that qubit by $i_k \in \{1,\dots,m_k\}$ and the corresponding distinct control states by $\vec{c}_{k, i_k} \in \vec{O}_{\backslash k}:= \{0,1\}^{n-1}$. We reserve $i_k = 0$ for the identity, setting $I_k:= \{0,\dots,m_k\}$. Consider a state in the computational basis, labeled by $\vec{o} \in \vec{O}:= \{0,1\}^n$. For each qubit $k$, the state $\vec{o}_{\backslash k}$ of the remaining qubits matches at most one control state $\vec{c}_{k, i_k}$. Therefore, $\vec{o}_{\backslash k}$ corresponds to a unique controlled unitary labeled by $i_k$, or to the identity ($i_k = 0$) if no control state matches. Consequently, each $\vec{o}$ is mapped to a unique global $\vec{i} = (i_1,\dots,i_n)$, and this mapping is a well-defined function, which we interpret as a quasi-process function $\omega:=(\omega_k:\vec{O}\rightarrow I_k)_{k=1}^n$. 
\end{proof}
Motivated by the connection discovered in Ref.~\cite{KB23} between the SHIFT ensemble and the AF/BW process, we are particularly interested in characterizing the set of unitaries that can be decomposed using process functions. These correspond to a special class of qPFUs, with the additional constraint that the quasi-process function in the decomposition must be a process function.
\begin{definition}[Process function unitary (PFU)]\label{def:PFU}
A qPFU is a PFU if the quasi-process function defined in its decomposition is a valid process function.
\end{definition}
We note that, although Theorem \ref{th:qPFU-decomposition} is specific to  multiqubit unitaries, the definition of qPFUs (Definition \ref{def:qPFU}) is general enough to accomodate unitaries acting on qudits, \textit{e.g.}, the unitary obtained from the two-qutrit domino basis \cite{BDF99} (which is a qPFU that is not a PFU) depicted in Figure \ref{fig:domino-circuit}. Its qPFU decomposition is the following:
\begin{align*}
    U^{(1)} &= \bigg((H \oplus 1)\otimes \ketbra{2}{2}+\id_1\otimes(\id_2-\ketbra{2}{2})\bigg)\\&\bigg((1 \oplus H)\otimes \ketbra{0}{0}+\id_1\otimes(\id_2-\ketbra{0}{0})\bigg)\\
    U^{(2)} &= \bigg(\ketbra{0}{0} \otimes (H \oplus 1)+(\id_1-\ketbra{0}{0})\otimes\id_2\bigg)\\&\bigg(\ketbra{2}{2} \otimes (1 \oplus H)+(\id_1-\ketbra{2}{2})\otimes\id_2\bigg)
\end{align*}
It is obtained from the quasi-process function $(i_1 = o_2 \oplus_3 2 , i_2 =  o_1 \oplus_{(2)} 1)$ ($\oplus_3$ is mod 3 addition and $\oplus_{(2)}$ is the regular mod 2 addition flipping 0 and 1 but leaving 2 invariant), with $i_k = 1$ corresponding to $H \oplus 1$ and $i_k = 2$ corresponding to $1 \oplus H$.

\begin{figure}[htb!]
    \centering
\begin{quantikz}
\lstick{$A$} & \gate{H \oplus 1} & \gate{1 \oplus H} & \ctrl{1} \wire[l][1]["\ket{0}"{above, pos=0.2}]{a}  &  \ctrl{1} \wire[l][1]["\ket{2}"{above, pos=0.2}]{a}  &\\
\lstick{$B$} & \ctrl{-1} \wire[l][1]["\ket{2}"{above, pos=0.2}]{a} &  \ctrl{-1} \wire[l][1]["\ket{0}"{above, pos=0.2}]{a}  & \gate{H \oplus 1} & \gate{1 \oplus H} &
\end{quantikz} 
    \caption{\label{fig:domino-circuit} Domino unitary circuit.}
\end{figure}
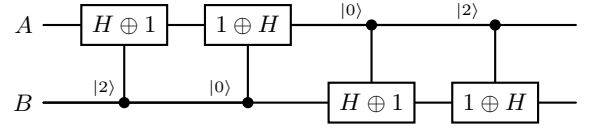

\section{Process function unitaries from mutually exclusive control conditions }\label{sec:mecc}

We define another family of unitaries, namely, those satisfying a constraint that we term \textit{mutually exclusive control conditions} (MECC).

\begin{definition}[Unitary with Mutually Exclusive Control Conditions (MECC-U)]\label{def:MECC-U}
Consider a unitary $U$ whose qPFU decomposition is such that the
controlled unitaries in Eq.~\ref{eq:local-U_PFU} are of the form
\begin{align}
	U_{\vec{c}_{k,i_k}}&:=\ketbra{\vec{c}_{k,i_k}}{\vec{c}_{k,i_k}}\otimes U^{(k)}_{i_k}+(\id_{\backslash k}-\ketbra{\vec{c}_{k,i_k}}{\vec{c}_{k,i_k}})\otimes
	\id_{k} \nonumber \\ 
    &\forall k\in\{1,2,\dots,n\},  i_k\in\{1,2,\dots,m_k\}, 
\end{align}
where $\vec{c}_{k,i_k}$ is the unique $\vec{o}_{\backslash k}$ such that $\omega_k(\vec{o}_{\backslash k}) = i_k$.

The control conditions $$\bigcup\limits_{k = 1}^n\bigcup\limits_{i_k=1}^{m_k}\{\vec{c}_{k, i_k}\}$$ are said to be mutually exclusive if they partition the set of strings $\vec{O}:=\bigtimes_{k=1}^nO_k$ (where $O_k=\{0,1,\dots,d_k-1\}$) into disjoint subsets given by 
\begin{align}
	&B_{k,i_k}:=\{\vec{o} | o_k\in O_k,\vec{o}_{\backslash k}=\vec{c}_{k,i_k}\},
\end{align}
where $k\in\{1,2,\dots,n\}$, $i_k\in\{1,2,\dots,m_k\}$, and 
\begin{align}
	&B_0:=\vec{O}\backslash \left(\bigsqcup_{k,i_k} B_{k,i_k}\right).
\end{align}
A unitary that admits such a decomposition satisfying MECC is said to be an MECC Unitary (MECC-U).

\end{definition}
To put it simply, given a computational basis state of the $n$ systems, $\ket{\vec{o}}$ (where $\vec{o}\in\vec{O}$), at most one of the unitaries in the decomposition acts non-trivially on the $n$-partite system. It is this property of the control conditions that we refer to as \textit{mutual exclusivity}, taking our cue from Ref.~\cite{NC06}. The Niset-Cerf unitaries correspond to a specific type of MECC circuits, that is, MECC circuits with at least one controlled unitary for each local system. These MECC circuits are what we call \textit{genuinely multipartite}.

\begin{definition}[Genuinely multipartite MECC-U]\label{def:genuine-MECC}
An MECC-U is genuinely multipartite if the control conditions are such that there is at least one non-empty subset $B_{k, i_k}$ for every party $k$.
\end{definition}
By definition, these circuits have a least $n$ layers. It is always possible to have a number of controlled unitaries strictly greater than $n$ while preserving the mutual exclusivity. One such example is given in Figure \ref{fig:MECC-circuit}, where $U_1$ and $U_2$ are two distinct unitaries. The question of the depth of MECC circuits is treated in Section \ref{sec:depth}. 

\begin{figure}[htb!]
    \centering
    \includegraphics[scale=0.8]{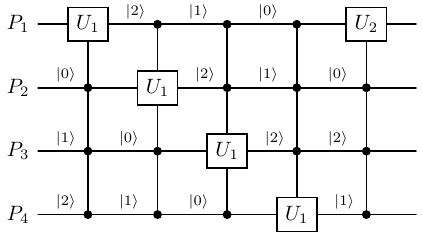}
    \caption{\label{fig:MECC-circuit} 4-party genuinely multipartite MECC unitary circuit with 5 controlled unitaries.}
\end{figure}

We now prove that every MECC-U is a PFU, \textit{i.e.}, mutual exclusivity of control conditions is a sufficient condition for a unitary to be a process function unitary. 

\begin{theorem} \label{th:meccpfu} 
All unitaries satisfying the mutual exclusivity of control conditions are process function unitaries.
\end{theorem}
\begin{proof}
Let $U$ be a MECC unitary acting on $n$ systems. Take the quasi-process $\omega := (\omega_k:\vec{O}_{\backslash k}\rightarrow I_k)_{k=1}^n$ defined by the structure of the given MECC-U, as described in Definition \ref{def:MECC-U}. We show that it satifies the fixed-point condition and thus that $U$ is in the set PFU.

The function $\omega$ takes values in the set $\Tilde{I}  = \big(\bigcup\limits_{k=1}^n\{i_k\vec{e}_{k} | i_k \in I_k \backslash \{0\}\}\big) \cup \vec{0} \subset \vec{I}$. 
This follows from the mutual exclusivity of the control conditions for different layers of the unitary $U$, \textit{i.e.}, for any $\vec{o}\in\vec{O}$, either exactly one of the parties applies a non-trivial local unitary ($\ket{\vec{o}} \in B_{k, i_k}\implies \vec{i}=i_k\vec{e}_k$, if this is the party $k$) or none of them do ($\ket{\vec{o}} \in B_0\implies\vec{i}=\vec{0}$). We can thus restrict ourselves to inputs of this form as these are the only possible fixed points when composing the function with local interventions. Now, for every choice of local interventions $h$, we can show that there is exactly one fixed point in this set. Indeed, for any choice of local operations acting on the set $\Tilde{I}$, since $(i_k\vec{e}_{k})_{\backslash k} = \vec{0}_{\backslash k} \forall k \in [n]$ and the operations $h$ are local, we have that  

\begin{equation}\label{eq}
h(i_k\vec{e}_{k})_{\backslash k} = h(\vec{0})_{\backslash k} \ \forall k \in [n].   
\end{equation}

In our construction, given a party $k$, the $k$-th component of $\omega$ only depends on the state of the systems of the remaining parties. From Eq.~\eqref{eq}, we can see that the vector $h(\vec{0})$ contains all the information about the control state for each party after a choice of local operations. If $h(\vec{0})_{\backslash k}$ does not correspond to a valid control state for any $k$, we have that $\omega_k(h(\vec{0})_{\backslash k}) = 0 \ \forall k$, and therefore $\vec{0}$ is the only fixed point. However, if it corresponds to a valid control state for a certain party $k'$ and choice of unitary $i_{k'} \in I_{k'}\backslash\{0\}$, we have that 
\begin{equation}\omega_{k'}(h(i_{k'}\vec{e}_{k'})_{\backslash k'}) = \omega_{k'}(h(\vec{0})_{\backslash k'}) = i_{k'}, \end{equation} where we used Eq.~\eqref{eq}. The exclusivity of the control conditions imposes  that all the remaining components of $w \circ h$ are null:
\begin{align}
\omega_k(h(i_k\vec{e}_{k})_{\backslash k}) &= \omega_k(h(\vec{0})_{\backslash k}) = 0 \\ 
&\forall i_k \in I_k, k \neq k'.    
\end{align}

Therefore, in this case, $i_{k'}\vec{e}_{k'}$ is the only fixed point.
\end{proof}

Theorem \ref{th:meccpfu} applies to the set of controlled unitary gates constructed by Niset-Cerf. Indeed, these unitaries are of the form $U = H^{d_1}_{\vec{c}_{\backslash 1}}...H^{d_n}_{\vec{c}_{\backslash n}}$, where $H^d$ is the $d$-dimensional Hadamard gate and the control conditions are cyclic permutations of the vector $(0, \dots, n-2)$. We provide the explicit form of the Niset-Cerf process functions in Table \ref{tab:nisetcerfPF}.

\begin{table}[]
    \centering
    \begin{tabular}{c|c}
       $\vec{i}$  &  $\vec{o}$\\
       \hline\\
       $\vec{e}_1$ & $\{(l_1,0,1,\dots,n-2)\}_{l_1=0}^{d_1-1}$\\
       $\vec{e}_2$ & $\{(n-2,l_2,0,1,\dots,n-3)\}_{l_2=0}^{d_2-1}$\\
       \vdots & \vdots\\
       $\vec{e}_k$ & $\{(n - k, n - k + 1, \dots,n-2,l_k,0,1,\dots,n - 1 -k)\}_{l_k=0}^{d_k-1}$\\
       \vdots & \vdots\\
       $\vec{e}_n$ & $\{(0,1,2,\dots,n-2,l_n)\}_{l_n=0}^{d_n-1}$\\
       $\vec{0}$  & $\vec{o}\in\vec{O}\backslash\vec{O}_{\rm ctrl}$  
    \end{tabular}
    \caption{Process function constructed from an $n$-qudit Niset-Cerf circuit, where qudit $j\in\{1,2,\dots,n\}$ has local dimension $d_j\geq n-1$. For $n=3$ and $d_1=d_2=d_3=2$, we recover the AF/BW process function corresponding to the SHIFT basis. The set $\vec{O}_{\rm ctrl}$ corresponds to valid control conditions.}
    \label{tab:nisetcerfPF}
\end{table}

For $n = 4$, one example of a process function corresponding to a unitary satisfying the mutual exclusivity condition is given by the generalization of the AF/BW process function found in M. Araújo et al. \cite{AGB17} (see Table \ref{tab:MECC-PF}).

\begin{table}[htb!]
    \centering
    \begin{tabular}{c|c}
$\vec{i}$ & $\vec{o}$ \\
        \hline
      $(0,0,0,0)$  &  $(0,0,0,0)$, $(0,1,0,1)$, $(0,1,1,1)$, $(1,0,1,0)$, \\
       & $(1,0,1,1)$, $(1,1,0,1)$, $(1,1,1,0)$, $(1,1,1,1)$  \\
      $(1,0,0,0)$  &  $(0,0,0,1)$, $(1,0,0,1)$ \\
      $(0,1,0,0)$  &  $(1,0,0,0)$, $(1,1,0,0)$ \\
      $(0,0,1,0)$  &  $(0,1,0,0)$, $(0,1,1,0)$ \\
      $(0,0,0,1)$  &  $(0,0,1,0)$, $(0,0,1,1)$ \\
    \end{tabular}
    \caption{4-party process function corresponding to a MECC unitary.}
    \label{tab:MECC-PF}
\end{table}

However, the converse of Theorem \ref{th:meccpfu} is not true, \textit{i.e.}, there exist PFUs that do not satisfy the mutual exclusivity of control conditions. 
An example of this type of PFU for $n = 4$ can be found by looking at another generalization of the AF/BW process from \"{A} Baumeler \textit{et al.}~\cite{BGR22} (see Table \ref{tab:non-MECC-PF}).

\begin{table}[htb!]
    \centering
    \begin{tabular}{c|c}
$\vec{i}$ & $\vec{o}$ \\
        \hline
      $(0,0,0,0)$  &  $(0,0,0,0)$, $(0,1,1,1)$ \\
      $(0,1,0,0)$  &  $(0,0,0,1)$, $(1,1,0,0)$ \\
      $(0,0,1,0)$  &  $(0,1,0,0)$, $(0,1,1,0)$ \\
      $(0,0,0,1)$  &  $(0,0,1,1)$, $(1,1,1,0)$ \\
      $(1,1,0,0)$  &  $(0,1,0,1)$, $(1,1,0,1)$ \\
      $(1,0,0,1)$  &  $(0,0,1,0)$, $(1,0,1,0)$ \\
      $(0,1,0,1)$  &  $(1,0,0,0)$, $(1,1,1,1)$ \\
      $(0,1,1,1)$  &  $(1,0,0,1)$, $(1,0,1,1)$ \\
    \end{tabular}
    \caption{4-party process function corresponding to a non-MECC unitary.}
    \label{tab:non-MECC-PF}
\end{table}
This process function can be used to construct a unitary without mutually exclusive control conditions (depicted in Figure \ref{fig:4party_non-mecc}). Indeed, we can choose to interpret each local input $i_k = 1$ as the application of a Hadamard. For each party, we list all the distinct vectors $\vec{o}_{\backslash k}$ such that $\omega_k(\vec{o}_{\backslash k}) = 1$. Then, we construct the unitary circuit with each of these vectors defining the control conditions for a single Hadamard. Here, we have that two or more parties can act at the same time. For example, looking at the fifth row from the top, we have that $i_1 = i_2 = 1$ since the vector $(0, 1, 0, 1)$ satisfies the control conditions $\vec{o}_{\backslash 1}=(1,0,1)$ and $\vec{o}_{\backslash 2}=(0,0,1)$ while the vector $(1, 1, 0, 1)$ satisfies $\vec{o}_{\backslash 1} = \vec{o}_{\backslash 2} = (1,0,1)$, for parties $1$ and $2$, respectively. 

\begin{figure}[htb!]
    \centering
    \begin{subfigure}[t]{0.2\textwidth}
    \includegraphics[scale=0.7]{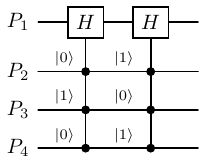}
    \caption{Party 1}  
    \end{subfigure}~
    \begin{subfigure}[t]{0.2\textwidth}
    \includegraphics[scale=0.7]{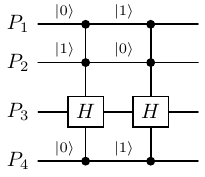}
    \caption{Party 3}
    \end{subfigure}
    
    \begin{subfigure}[b]{0.2\textwidth}
    \includegraphics[scale=0.7]{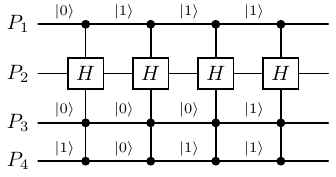}
    \caption{Party 2}
    \end{subfigure}~
    \begin{subfigure}[b]{0.25\textwidth}
    \includegraphics[scale=0.7]{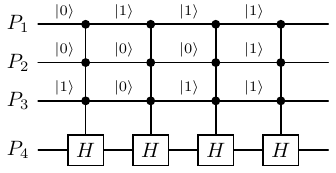}
    \caption{Party 4}
    \end{subfigure}
    \caption{Non-MECC circuit corresponding to the 4-party process function in Table \ref{tab:non-MECC-PF}.}
\label{fig:4party_non-mecc} 
\end{figure}

\section{Depth of a MECC circuit}\label{sec:depth}

Looking at the definition of MECC unitaries, one might wonder about how to characterize the depth of such circuits. For example, it was noted by Niset and Cerf that circuits with more layers would generate bases exhibiting more QNLWE \cite{NC06}. Additionally, in light of the relationship between QNLWE and non-causal process functions (\cite{KB23}, \cite{DSL25}), this depth might tell us something about the causal structure of the process function in the PFU decomposition.

\begin{figure}[htb!]
    \centering
    \includegraphics[scale=0.6]{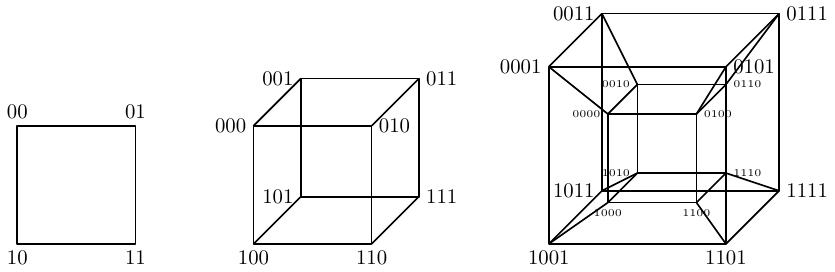}
     \caption{For $n$-qubits, the generalized Hamming graph corresponds to the hypercube graph, depicted from $n = 2$ to $n = 4$.}
     \label{fig:hypercubes}
\end{figure}

It turns out that the task of finding mutually exclusive control conditions can be recast into a combinatorial problem in graph theory. Given $n$ parties holding a composite system of dimension $d = d_1 \cdots d_n$, ensuring MECC corresponds to grouping the strings in $\vec{O} = \{0, ...,d_1-1\}\times\dots\times\{0, ...,d_n-1\}$ into disjoint subsets (these correspond to the subsets $B_{k,i_k}$ in Definition \ref{def:MECC-U}) such that the number of positions any two strings in a given subset differ by (\textit{i.e.}, their Hamming distance) is exactly one. 

We can define a graph such that the vertices are the strings in $\vec{O}$ and two vertices are connected by an edge if and only if their Hamming distance is one. This graph belongs to the class of generalized Hamming graphs.

\begin{definition}[Generalized Hamming graph]\label{def:hamming}
A generalized Hamming graph $H(n, [d_1,\dots,d_n])$ is a graph such that the set of vertices are the $n$-tuples in the mixed alphabet $\vec{O} = \bigtimes_{k=1}^n \{0, ..., d_k-1\}$ and two vertices $\vec{v}, \vec{w} \in \vec{O}$ are connected by an edge if and only if there is exactly one position $k \in [n]$ for which $v_k \neq w_k$.
\end{definition} 

Given a party $k$, one control condition is equivalent to a set of vertices in $H(n, [d_1,\dots,d_n])$, all adjacent to each other, also known as a \textit{clique}.

\begin{definition}[Clique]
Given a graph $G = (V, E)$, a clique $C \subseteq V$ is a set of vertices such that every two distinct vertices are adjacent 
\end{definition}
Since any clique in $H(n, [d_1,\dots,d_n])$ contains vertices differing in a single position $k$, such a clique contains at most $d_k$ vertices. Consequently, a control condition is precisely a clique of size $d_k$, i.e. a maximal clique in direction $k$. It is then easy to see that finding mutually exclusive control conditions is equivalent to finding a set of disjoint maximal cliques in this class of graphs. Restricting ourselves to genuinely multipartite MECC (see Definition \ref{def:genuine-MECC}), we consider sets of disjoint maximal cliques such that, for every position $k \in \{1, ...,n\}$, there is at least one clique composed of vertices differing only in that position.  

\subsection{Minimum depth}
The minimum depth of a genuinely multipartite MECC circuit is $l_{min} = n$. In that case, every party's choice of a local basis is controlled by the state of the remaining parties. The associated process function is such that no component $\omega_k$ is a constant function and is thus a process function \textit{without global past} \cite{KB23}.

\begin{theorem}
Let $U$ be a MECC-U acting on $n$ systems, then $U$ is a genuinely multipartite MECC-U if and only if its PFU decomposition is given by a process function without global past. 
\end{theorem}
\begin{proof}
We first prove that process functions without global past imply genuinely multipartite MECC and then the converse.
\begin{enumerate}
\item Let $\omega$ be a process function without global past. By definition, no global past means that every party $k$ receives a signal from at least some other party $k'$ through $\omega$, i.e, 
\begin{align}\label{eq:no-global-past}
&\forall k \ \exists k'\neq k, o_{k'}, \tilde{o}_{k'}: \nonumber \\ &\omega_k(o_{k'}, \vec{o}_{\backslash \{k,k'\}}) \neq  \omega_k(\tilde{o}_{k'}, \vec{o}_{\backslash \{k,k'\}}).   
\end{align}
Therefore, at least one of $(o_{k'}, \vec{o}_{\backslash \{k,k'\}})$ or $(\tilde{o}_{k'}, \vec{o}_{\backslash \{k,k'\}})$ is a valid control condition corresponding to a unique $i_k \neq 0$. Consequently, there is at least one non-empty $B_{k, i_k}$ for every party; $U$ is a genuinely multipartite MECC-U.
    
    \item Let $U$ be a genuinely multipartite MECC. This implies that every party changes their basis at least once, controlled on the state of the others. Therefore, every component of the process function in the PFU decomposition takes at least two values and the definition of no global past (Eq.~\ref{eq:no-global-past}) is satisfied.
\end{enumerate}
\end{proof}

\subsection{Maximum depth}
To find the maximum depth of MECC circuits, we consider sets of maximal cliques which cover all the vertices, i.e. a partition. The maximum depth is then given by maximizing the number of cliques in the partition. When $d_k = 2 \ \forall k$,  $H(n, [d_1,\dots,d_n])$ corresponds to the $n$-dimensional hypercube graph $Q_n$ (see Figure \ref{fig:hypercubes}). In that case, the cliques of interest are edges and the partitioning problem is known as the \textit{perfect matching} problem.
\begin{definition}[Perfect matching in a graph]
Given a graph $G = (V, E)$, a perfect matching $M \subseteq V$ is a set of edges such that every vertex in $V$ is adjacent to a unique edge in $M$.
\end{definition}
We now return our attention to genuinely multipartite MECC unitaries. The perfect matching problem, when restricted to genuinely multipartite matchings, has been solved in Ref.~\cite{matchings} for hypercubes: for $n > 3$, one can always find a genuinely multipartite perfect matching in $Q_n$. Since a perfect matching in $Q_n$ contains exactly $2^{n-1}$ edges, we immediately have the following result:
\begin{theorem}\label{th:max-depth-qubit}
For $n > 3$ the maximum depth of a genuinely multipartite $n$-qubit MECC circuit is $l_{max} = 2^{n -1}$.
\end{theorem}
A simple example achieving this maximum depth for $n = 4$ qubits is given by the process function in Table \ref{tab:perf-MECC-PF} and its corresponding MECC circuit in Figure \ref{fig:4party_perf-mecc}.
\begin{table}[htb!]
    \centering
    \begin{tabular}{c|c}
$\vec{i}$ & $\vec{o}$ \\
        \hline
      $(1,0,0,0)$  & $(0,1,1,0)$, $(1,1,1,0)$ \\
      $(2,0,0,0)$  & $(0,1,1,1)$, $(1,1,1,1)$ \\
      $(0,1,0,0)$  & $(1,0,0,0)$, $(1,1,0,0)$ \\
      $(0,2,0,0)$  & $(1,0,0,1)$, $(1,1,0,1)$ \\
      $(0,0,1,0)$  & $(0,0,0,0)$, $(0,0,1,0)$ \\ 
      $(0,0,2,0)$  & $(0,0,0,1)$, $(0,0,1,1)$ \\
      $(0,0,0,1)$  & $(0,1,0,0)$, $(0,1,0,1)$ \\ 
      $(0,0,0,2)$  & $(1,0,1,0)$, $(1,0,1,1)$
    \end{tabular}
    \caption{4-party process function corresponding to a maximum depth MECC unitary. See Fig.~\ref{fig:4party_perf-mecc} for the structure of the unitary.}
    \label{tab:perf-MECC-PF}
\end{table}

\begin{figure}[htb!]
    \centering
    \begin{subfigure}[h]{0.2\textwidth}
    \includegraphics[scale=0.7]{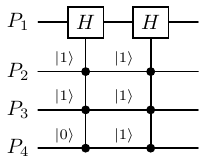}
    \caption{Party 1}  
    \end{subfigure}~
    \begin{subfigure}[h]{0.2\textwidth}
    \includegraphics[scale=0.7]{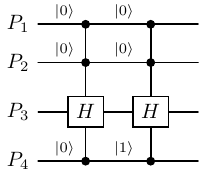}
    \caption{Party 3}
    \end{subfigure}
    
    \begin{subfigure}[h]{0.2\textwidth}
    \includegraphics[scale=0.7]{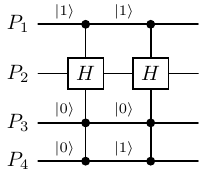}
    \caption{Party 2}
    \end{subfigure}~
    \begin{subfigure}[h]{0.2\textwidth}
    \includegraphics[scale=0.7]{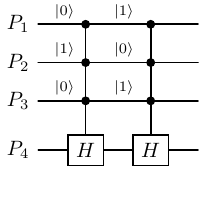}
    \caption{Party 4}
    \end{subfigure}
    \caption{Maximum depth MECC circuit corresponding to the 4-party process function in Table \ref{tab:perf-MECC-PF}.}
\label{fig:4party_perf-mecc} 
\end{figure}
 We conjecture that the result of Ref.~\cite{matchings} can be generalized to arbitrary $[d_1,\ldots,d_n]$, with $d_k \geq 2$. Note that $n  > 3$ is necessary in general as it can be verified that genuinely multipartite partitioning is impossible for $n \leq 3$ for any choice of $[d_1, \dots, d_n]$.
\begin{conjecture}
For $n > 3$, there exists a genuinely multipartite partitioning into maximal cliques of $H(n, [d_1, . . . , d_n])$.
\end{conjecture}

\section{Mutual exclusivity versus unambiguity}\label{sec:meccunamb}
In this section, we are interested in the logical relationship between MECC and \textit{unambiguity}, a property of product bases defined by H. Dourdent \textit{et al.}~\cite{DSL25} and which is shown to be equivalent to the fixed-point condition of process functions in certain cases.

Consider an arbitrary product basis $B = \{\ket{\psi^j} = \ket{\psi_1^j}\otimes\dots\otimes\ket{\psi_n^j}\}_{j=1}^{|B|}$. For each party $k \in [n]$, take the set of distinct local states  $S_k = \{\ket{\psi^j_k}\}_{j=1}^{s_k \leq |B|}$. This local set can be divided into $m_k$ local orthonormal bases $S_{i_k}$: $S_k = \bigcup\limits_{i_k = 0}^{m_k - 1}S_{i_k} \forall k\in[n]$. We can label the global states $\ket{\psi^j}$ in the basis as $(\vec{o}^j|\vec{i}^j)$, where $i^j_k \in \{0, ..., m_k-1\}$ is the index of the local basis $S_{i^j_k}$ to which the local state $\ket{\psi^j_k}$ belongs and $o^j_k \in \{0,\dots,d_k -1\}$ is the position of $\ket{\psi^j_k}$ in that local basis. Unambiguous product bases are such that each local vector is assigned a unique local label $(o^j_k|i^j_k)$. 

\begin{definition}[Unambiguous product basis \cite{DSL25}]
An $n$-partite orthonormal product basis $B = \{\ket{\psi^j} = \ket{\psi_1^j}\otimes\dots\otimes\ket{\psi_n^j}\}_{j=1}^{|B|}$ is said to be unambiguous if, for every party $k$, two distinct local vectors $\ket{\psi_k^j} \neq \ket{\psi_k^{j'}}$ are orthogonal if and only if they belong to the same local basis $S_{i_k}$: 
\begin{equation}\label{def:unambiguity}
\braket{\psi_k^j}{\psi_k^{j'}}= 0 \iff \exists i_k \textrm{ such that } \ket{\psi_k^j},\ket{\psi_k^{j'}}\in S_{i_k}.
\end{equation}
\end{definition}    
Note that unambiguity implies the disjointness of the local bases $S_{i_k}$ for every party. The equivalence between unambiguous bases and process functions is due to the fact that the global labels $(\vec{o}^j|\vec{i}^j)$, if seen as a mapping from outputs $\vec{O}$ to inputs $\vec{I}$, define a valid process function \cite{DSL25}.

One example of an unambiguous basis is the SHIFT basis as each party chooses between the disjoint bases $\{\ket{0}, \ket{1}\}$ or $\{\ket{+}, \ket{-}\}$. This corresponds to the set of labels $\{(000|000), (000|111), (100|001), (100|101),$ $(010|100), (010|110), (001|010), (001|011)\}$, which is exactly the AF/BW process function. The SHIFT unitary is thus what we term an \textit{unambiguous unitary}.

\begin{definition}[Unambiguous unitary (Un-U)]
A unitary $U$ is an unambiguous unitary if it rotates the computational basis into an unambiguous product basis.
\end{definition}

Since unambiguous bases correspond to process functions, it is natural to expect that unambiguous unitaries are in the set PFU. We show this below.

\begin{theorem}
Every unambiguous unitary is a PFU.\end{theorem}
\begin{proof}
Given an unambiguous product basis $B = \{\ket{\psi^j} = \ket{\psi_1^j} \otimes\dots\otimes\ket{\psi_n^j}\}_{j=1}^{|B|}$, the labels $(\vec{o}^j|\vec{i}^j)$ as described earlier define a valid process function $\omega : \vec{O} \rightarrow \vec{I}$ \cite{DSL25}. In particular, we can rewrite the basis as 
\begin{equation}\label{eq:unambiguous-unitary}
B = \left\{\bigotimes\limits_{k = 1}^n U^{(k)}_{\omega_k(\vec{o}_{\backslash k})}\ket{\vec{o}}\Big| \vec{o} \in \vec{O}\right\},
\end{equation}
where $U^{(k)}_{\omega_k(\vec{o}_{\backslash k}) =: i_k}$ is the local unitary rotating the local computational basis of party $k$ (which is labeled by $i_k=0$) to the local basis $S_{i_k}$. The unitary rotating the computational basis to the unambiguous basis $B$ thus admits a PFU decomposition with the local unitaries in Eq.~\ref{eq:local-U_PFU} chosen to be the $U^{(k)}_{\omega_k(\vec{o}_{\backslash k})=:i_k}$ we just defined.
\end{proof}

Now, we show that MECC and unambiguity do not imply each other by constructing simple counter-examples. 
\begin{figure}[htb!]
    \centering
    \begin{subfigure}[b]{0.5\textwidth}
        \centering
        \includegraphics[scale=0.9]{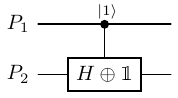}   
        \caption{Ambiguous and MECC unitary}
        \label{fig:ambiguous_circuit}
     \end{subfigure}
    \begin{subfigure}[b]{0.25\textwidth}
        \centering
        \includegraphics[scale=0.8]{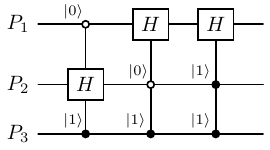}   
        \caption{Unambiguous non-MECC unitary}
        \label{fig:non-mecc_circuit}
     \end{subfigure}~
     \begin{subfigure}[b]{0.25\textwidth}
        \centering
        \includegraphics[scale=0.75]{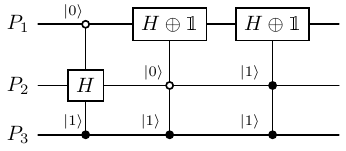}  \caption{Ambiguous and non-MECC unitary}
        \label{fig:non-mecc_ambiguous_circuit}
     \end{subfigure}
    \caption{Three distinct types of PFUs.}
\end{figure}

\begin{theorem}
There are MECC unitaries that are not unambiguous unitaries and vice-versa.
\end{theorem}

\begin{proof}
\begin{enumerate}
\item MECC $\nRightarrow$ unambiguity:
    The product basis $\{\ket{00},\ket{01}, \ket{02}, \ket{1+},\ket{1-}, \ket{12}\}$ is ambiguous as the second party's local states belong to the two overlapping bases $\{\ket{0}, \ket{1}, \ket{2}\}$ and $\{\ket{+}, \ket{-}, \ket{2}\}$.
    
However, this basis can be generated from the trivially MECC circuit represented in Figure \ref{fig:ambiguous_circuit}, where the second party's Hadamard acts on the qubit subspace $\{\ket{0}, \ket{1}\}$ and leaves the state $\ket{2}$ invariant. This unitary corresponds to the causal process function $(i_1 = 0, i_2 = o_1)$ from outputs $\{0, 1\} \times \{0, 1, 2\}$ to inputs $\{0\} \times \{0, 1\}$.

 \item  Unambiguity $\nRightarrow$ MECC: Consider the unambiguous product basis 
    \begin{align}\label{eq:non-mecc_basis}
    \begin{split}
    &\{\ket{000}, \ket{++1}, \ket{010}, \ket{+-1}, \\
    &\ket{100}, \ket{-01}, \ket{110}, \ket{-11}\}.  
    \end{split}   
    \end{align}
The corresponding unitary circuit is given in Figure \ref{fig:non-mecc_circuit}. This circuit is not MECC as the input state $\ket{0, 0, 1}$ activates the Hadamard gate for both parties, $P_1$ and $P_2$. It corresponds to the causal process function $(i_1 = o_3, i_2 = (o_1\oplus 1)o_3, i_3 = 0)$ from outputs $\{0, 1\} \times \{0, 1\} \times \{0, 1\}$ to inputs $\{0, 1\} \times \{0, 1\} \times \{0\}$.
\end{enumerate} 
\end{proof}

\begin{theorem}
There are process function unitaries that are neither MECC nor unambiguous unitaries. 
\end{theorem} 
\begin{proof}
Consider the following product basis
\begin{align}
\begin{split}
    &\{\ket{000}, \ket{++1}, \ket{010}, \ket{+-1}, \ket{100}, \ket{-01}, \\
    &\ket{110}, \ket{-11}, \ket{200}, \ket{201}, \ket{210}, \ket{211}\}.    
\end{split}
\end{align}
This product basis includes the states of the one defined in Eq.~\eqref{eq:non-mecc_basis}, with the difference that the first party holds a qutrit instead of a qubit. It can be generated from a circuit almost identical to the non-MECC circuit (Figure \ref{fig:non-mecc_circuit}), where this time the first party's Hadamard acts only on the qubit subspace $\{\ket{0}, \ket{1}\}$ and leaves the state $\ket{2}$ invariant. It is thus ambiguous since that party's local vectors belong to the intersecting bases $\{\ket{0}, \ket{1}, \ket{2}\}$ and $\{\ket{+}, \ket{-}, \ket{2}\}$. In spite of that, it corresponds to a causal process function, $(i_1 = o_3, i_2 = (o_1\oplus 1)o_3, i_3 = 0)$, with the same algebraic expression as in the case of Figure \ref{fig:non-mecc_circuit} but from outputs $\{0, 1\} \times \{0, 1, 2\} \times \{0, 1\}$ to inputs $\{0, 1\} \times \{0, 1\} \times \{0\}$. The unitary is depicted in Figure \ref{fig:non-mecc_ambiguous_circuit}.
\end{proof}

\section{Unitary Processes from Process function unitaries}\label{sec:uniprocess}
Given that process functions originally arose from considerations of indefinite causal order in the process matrix framework \cite{OCB12,BW16}, a natural question that arises is whether process function unitaries teach us anything about the implementation of the process function as a unitary process. Although we do not have a general construction for arbitrary PFUs, here we show how, given a MECC unitary (which is a PFU, \textit{cf.}~Theorem \ref{th:meccpfu}), we can construct a unitary purification of the associated process function in the process-matrix framework. Conceptually, this proceeds in two steps:

\begin{algorithm}[H]
\renewcommand{\thealgocf}{}
\caption{Unitary process 
from MECC-U}\label{protocol}
\begin{enumerate}
    \item Using the MECC-U, extract the process function using the protocol outlined in Ref.~\cite{KB23} for the case of the SHIFT measurement.
    
    \item Using the extracted process function, construct the associated unitary process as described in Ref.~\cite{AFN17}.
\end{enumerate}
\end{algorithm}

We illustrate this procedure for the case of the SHIFT PFU in Figure \ref{fig:unitarypurification} and prove that a similar construction would work for any MECC unitary under some restriction on the local unitaries. This is motivated by the fact that any process function can be extended into a reversible process function (by adding additional parties, see Theorem \ref{th:rev-pf} in Appendix \ref{sec:appendix-a}, adapted from Ref.~\cite{B17}) and that every reversible (process) function can be encoded into a unitary operator (see Corollary \ref{cor:unitary-pf}), albeit one which may not correspond to a valid unitary process.

First, we recall the extraction protocol described in Ref.~\cite{KB23} and give a necessary condition for its implementation. Consider an arbitrary product basis $B = \{\ket{\psi^j} = \ket{\psi_1^j}\otimes\dots\otimes\ket{\psi_n^j}\}_{j=1}^{|B|}$ and label the states as described in the beginning of Section \ref{sec:meccunamb}. Now, assume that the parties locally generate a state in the computational basis (labeled by $\vec{o}\in\vec{O}:=\{0,1\}^n$), feed it to a device measuring in the basis $B$, and then receive from this device a classical description of the post-measurement state. In other words, if the measured state is $\ket{\psi^j} = \bigotimes\limits_{k = 1}^n \ket{\psi^j_k}$, each party receives the label $\psi^j_k$. We want to construct local functions $f_k : L_k \rightarrow I_k$ that associate each description of a local vector to a unique local basis $i_k\in I_k$, with $L_k = \{\psi^j_k\}_{j=1}^{s_k \leq |B|}$ being the set of the local descriptions of party $k$'s distinct local vectors. 

Suppose that each local vector $\ket{\psi^j_k}$ is uniquely associated to a label $(o_k^j|i_k^j)$. This is only possible if the local bases $S_{i_k}$ are disjoint. In that case, an obvious choice of $f_k$ is therefore a function defined via the local labels $(o_k^j|i_k^j)$, \textit{i.e.}, if party $k$ measured $\ket{\psi^j_k}$, the function $f_k$ will output the same $i^j_k$ as specified by the corresponding local label $(o_k^j|i_k^j)$.

\begin{lemma}
Consider an n-partite product basis $B$ such that for every $k \in [n]$ and any two distinct $i_k \neq i'_k$:
\begin{equation}\label{eq:disjoint-local-bases}
    S_{i_k} \cap S_{i'_k}  = \emptyset. 
\end{equation}
Then the local labels $(o_k^j|i_k^j)$ define valid local functions $f_k : \psi^j_k \mapsto i_k^j$.
\end{lemma}

\begin{proof}
Since the local bases are disjoint, any local vector $\ket{\psi^j_k}$ is uniquely identified by some label $(o_k^j|i_k^j)$. Therefore, every party $k$ can uniquely assign the basis label $i_k^j$ to the local vector's classical  description $\psi^j_k$, defining a valid function function $f_k : \psi^j_k \mapsto i_k^j$.
\end{proof}
Now, consider a product basis that is generated by a PFU. It is of the form $B =\left\{\bigotimes\limits_{k = 1}^n U^{(k)}_{i_k = \omega_k(\vec{o}_{\backslash k})}\ket{\vec{o}} | \vec{o} \in \vec{O}\right\}$, where $U^{(k)}_{i_k}$ are the local unitaries in Eq.~\eqref{eq:local-U_PFU}. Assuming that this basis satisfies Eq.~\eqref{eq:disjoint-local-bases}, we obtain the following equivalent property of the local unitaries as a corollary:

\begin{corollary}\label{cor:local-unitaries}
Consider an n-partite product basis $B$ generated by a PFU $U$ such that for every $k \in [n]$ and any two distinct $i_k \neq i'_k$:
\begin{equation}\label{eq:unitaries-disjoint-local-bases}
    U^{(k)}_{i_k} \ket{o_k} \neq U^{(k)}_{i'_k} \ket{o_k} \ \forall o_k \in O_k.
\end{equation}
Then the local labels $(o_k^j|i_k^j)$ define valid local functions $f_k : \psi^j_k \mapsto i_k^j$.
\end{corollary}
We can now state the main theorem of this section:
\begin{theorem}\label{th:purification}
Given a MECC-U such that its local unitaries (in the PFU decomposition) satisfy Eq.~\eqref{eq:unitaries-disjoint-local-bases}, the steps described in \hyperref[protocol]{Protocol} provide a unitary process purifying the associated process function.
\end{theorem}
The proof of this theorem can be found in Appendix \ref{sec:appendix-a}. It is important to note that Corollary \ref{cor:local-unitaries} gives a necessary but insufficient condition for a successful extraction protocol. Nonetheless, as explained in the proof, the additional assumption that the PFU is MECC ensures the success of the first step of \hyperref[protocol]{Protocol}. It also ensures that the unitary operator constructed in the second step of \hyperref[protocol]{Protocol} is a valid process matrix. 

After numerically verifying the validity for a wide range of process functions without the MECC property, we conjecture that the purification step (the second step of \hyperref[protocol]{Protocol}) generalizes to all process functions, \textit{i.e.},
\begin{conjecture}
Any process function can be purified to a unitary process.
\end{conjecture}

\begin{figure*}[htb!]
    \centering
    \includegraphics[scale=0.4]{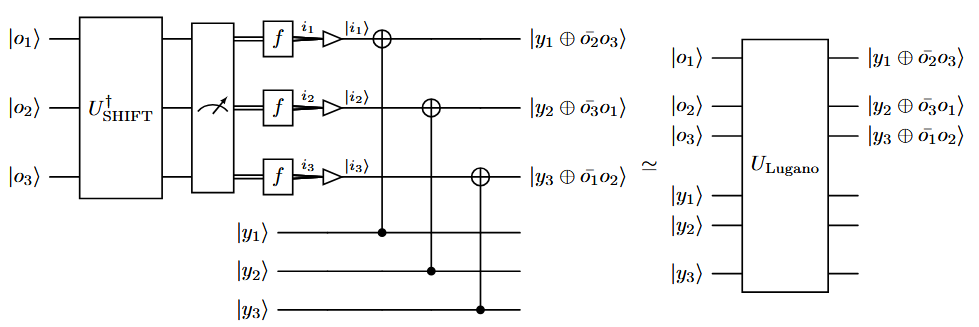}
    \caption{Purification of the AF/BW process function via the PFU of Figure \ref{fig:SHIFT-circuit} in the following steps: 1) Use the implementation of the SHIFT basis measurement as in the protocol of Figure \ref{fig:SHIFT-PFU}, 2) Prepare qubit states depending on the local observed outcomes, \textit{i.e.}, for outcomes $\{0,+\}$, we prepare the state $\ket{0}$ and for outcomes $\{1,-\}$, we prepare the state $\ket{1}$ (the box ``$f$" denotes the bit encoding $f : 0 \rightarrow 0, 1 \rightarrow 0, + \rightarrow 1, - \rightarrow 1$ and is followed by the encoding of this bit in the computational basis of a qubit), and 3) Use the standard purification of the AF/BW process function (\textit{i.e.}, by first making the function reversible and then constructing a unitary implementation), as outlined in Ref.~\cite{AFN17}.}
    \label{fig:unitarypurification}
\end{figure*}

\section{Discussion}\label{sec:conclusion}
We have introduced the quasi-process function representation of unitaries,  motivated by the profound ties between QNLWE ensembles and indefinite causal order \cite{KB23, DSL25}. We used this representation to investigate the structure of the unitaries which map the computational basis to QNLWE bases that can be measured using a process function \cite{KB23,DSL25}. In particular, we asked: Which unitaries are process function unitaries (PFUs)? 

Inspired by the unitary-based QNLWE construction of Niset and Cerf \cite{NC06}, we defined a class of unitaries, namely unitaries with mutually exclusive control conditions (MECC-Us). We showed that this property of the control conditions is sufficient for a unitary to be a process function unitary. We also reformulated the problem of constructing MECC unitaries using properties of generalized Hamming graphs, obtaining combinatorial insights into the circuit depth of such unitaries. 

We then considered unambiguous product bases \cite{DSL25}, another class of bases that are linked to process functions, and their corresponding unambiguous unitaries (Un-Us). We proved that unambiguity and MECC are two inequivalent properties. In fact, unambiguity concerns the local states (in the way they can be partitioned) while MECC is a property of the global states (how at most one party can be in a rotated basis). We constructed a simple example demonstrating that while both MECC and unambiguity are sufficient for encoding a process function, neither is necessary. Finally, we obtained a protocol to extract the AF/BW process function from the SHIFT unitary in order to then obtain a unitary purification of the process function as described in Ref.~\cite{AFN17}. We proved that this protocol can be generalized to any MECC unitary, thereby providing a systematic way to generate unitary processes.

These results offer several avenues for future work. It remains to find a condition that is both sufficient and necessary for a qPFU to be a PFU. It would also be interesting to investigate the relationship between MECC and unambiguity when restricted to non-causal process functions. Although this work is focused on deterministic examples of indefinite causal order, there also exist intrinsically probabilistic classical processes that are noncausal \cite{BFW14}. Ref.~\cite{A22} shows how such a process, the Baumeler-Feix-Wolf (BFW) process \cite{BFW14}, can be used to construct a separable operation that cannot be implemented by finite-round LOCC. A possible extension would therefore be to go beyond unitaries to study separable operations and their relationship to classical processes in general. 

We end with a final comment on the qPFU representation. The fact that all $n$-qubit unitaries are qPFUs suggests a possible new perspective on quantum computation in the circuit model. Quasi-process functions in this representation capture how the the control conditions across the qubits interact without fixing the local gates acting on the qubits. They offer an avenue for tools from the study of indefinite causal order to find applications in quantum information processing beyond the case of QNLWE \cite{KB23,DSL25}. A concrete example of how qPFUs could be useful beyond the application domain discussed here is ongoing work where we use their properties to investigate the structure of non-Clifford unitaries \cite{BK04}.

\section*{Acknowledgments}
The authors thank \"Amin Baumeler,  Cyril Branciard, Pablo Concha-Vega, Ognyan Oreshkov, Leonardo Vaglini, and Vilasini Venkatesh for insightful discussions. This work received support from the French government under the France 2030 investment plan, as part of the Initiative d'Excellence d'Aix-Marseille Université-A*MIDEX, AMX-22-CEI-01. NDA acknowledges financial support from the Program QuanTEdu-France
(Grant No. ANR-22-CMAS-0001 France 2030).
\bibliography{mainNotes}

\begin{thebibliography}{28}%
\makeatletter
\providecommand \@ifxundefined [1]{%
 \@ifx{#1\undefined}
}%
\providecommand \@ifnum [1]{%
 \ifnum #1\expandafter \@firstoftwo
 \else \expandafter \@secondoftwo
 \fi
}%
\providecommand \@ifx [1]{%
 \ifx #1\expandafter \@firstoftwo
 \else \expandafter \@secondoftwo
 \fi
}%
\providecommand \natexlab [1]{#1}%
\providecommand \enquote  [1]{``#1''}%
\providecommand \bibnamefont  [1]{#1}%
\providecommand \bibfnamefont [1]{#1}%
\providecommand \citenamefont [1]{#1}%
\providecommand \href@noop [0]{\@secondoftwo}%
\providecommand \href [0]{\begingroup \@sanitize@url \@href}%
\providecommand \@href[1]{\@@startlink{#1}\@@href}%
\providecommand \@@href[1]{\endgroup#1\@@endlink}%
\providecommand \@sanitize@url [0]{\catcode `\\12\catcode `\$12\catcode `\&12\catcode `\#12\catcode `\^12\catcode `\_12\catcode `\%12\relax}%
\providecommand \@@startlink[1]{}%
\providecommand \@@endlink[0]{}%
\providecommand \url  [0]{\begingroup\@sanitize@url \@url }%
\providecommand \@url [1]{\endgroup\@href {#1}{\urlprefix }}%
\providecommand \urlprefix  [0]{URL }%
\providecommand \Eprint [0]{\href }%
\providecommand \doibase [0]{https://doi.org/}%
\providecommand \selectlanguage [0]{\@gobble}%
\providecommand \bibinfo  [0]{\@secondoftwo}%
\providecommand \bibfield  [0]{\@secondoftwo}%
\providecommand \translation [1]{[#1]}%
\providecommand \BibitemOpen [0]{}%
\providecommand \bibitemStop [0]{}%
\providecommand \bibitemNoStop [0]{.\EOS\space}%
\providecommand \EOS [0]{\spacefactor3000\relax}%
\providecommand \BibitemShut  [1]{\csname bibitem#1\endcsname}%
\let\auto@bib@innerbib\@empty
\bibitem [{\citenamefont {Kunjwal}\ and\ \citenamefont {Oreshkov}(2023)}]{KO23}%
  \BibitemOpen
  \bibfield  {author} {\bibinfo {author} {\bibfnamefont {R.}~\bibnamefont {Kunjwal}}\ and\ \bibinfo {author} {\bibfnamefont {O.}~\bibnamefont {Oreshkov}},\ }\bibfield  {title} {\bibinfo {title} {Nonclassicality in correlations without causal order},\ }\href {https://arxiv.org/abs/2307.02565} {\bibfield  {journal} {\bibinfo  {journal} {arXiv preprint arXiv:2307.02565}\ } (\bibinfo {year} {2023})}\BibitemShut {NoStop}%
\bibitem [{\citenamefont {Bennett}\ \emph {et~al.}(1999{\natexlab{a}})\citenamefont {Bennett}, \citenamefont {DiVincenzo}, \citenamefont {Fuchs}, \citenamefont {Mor}, \citenamefont {Rains}, \citenamefont {Shor}, \citenamefont {Smolin},\ and\ \citenamefont {Wootters}}]{BDF99}%
  \BibitemOpen
  \bibfield  {author} {\bibinfo {author} {\bibfnamefont {C.~H.}\ \bibnamefont {Bennett}}, \bibinfo {author} {\bibfnamefont {D.~P.}\ \bibnamefont {DiVincenzo}}, \bibinfo {author} {\bibfnamefont {C.~A.}\ \bibnamefont {Fuchs}}, \bibinfo {author} {\bibfnamefont {T.}~\bibnamefont {Mor}}, \bibinfo {author} {\bibfnamefont {E.}~\bibnamefont {Rains}}, \bibinfo {author} {\bibfnamefont {P.~W.}\ \bibnamefont {Shor}}, \bibinfo {author} {\bibfnamefont {J.~A.}\ \bibnamefont {Smolin}},\ and\ \bibinfo {author} {\bibfnamefont {W.~K.}\ \bibnamefont {Wootters}},\ }\bibfield  {title} {\bibinfo {title} {Quantum nonlocality without entanglement},\ }\href {https://doi.org/10.1103/PhysRevA.59.1070} {\bibfield  {journal} {\bibinfo  {journal} {Phys. Rev. A}\ }\textbf {\bibinfo {volume} {59}},\ \bibinfo {pages} {1070} (\bibinfo {year} {1999}{\natexlab{a}})}\BibitemShut {NoStop}%
\bibitem [{\citenamefont {Baumeler}\ and\ \citenamefont {Wolf}(2016{\natexlab{a}})}]{BW16}%
  \BibitemOpen
  \bibfield  {author} {\bibinfo {author} {\bibfnamefont {{\"A}.}~\bibnamefont {Baumeler}}\ and\ \bibinfo {author} {\bibfnamefont {S.}~\bibnamefont {Wolf}},\ }\bibfield  {title} {\bibinfo {title} {The space of logically consistent classical processes without causal order},\ }\href {https://doi.org/10.1088/1367-2630/18/1/013036} {\bibfield  {journal} {\bibinfo  {journal} {New Journal of Physics}\ }\textbf {\bibinfo {volume} {18}},\ \bibinfo {pages} {013036} (\bibinfo {year} {2016}{\natexlab{a}})}\BibitemShut {NoStop}%
\bibitem [{\citenamefont {Kunjwal}\ and\ \citenamefont {Baumeler}(2023)}]{KB23}%
  \BibitemOpen
  \bibfield  {author} {\bibinfo {author} {\bibfnamefont {R.}~\bibnamefont {Kunjwal}}\ and\ \bibinfo {author} {\bibfnamefont {A.}~\bibnamefont {Baumeler}},\ }\bibfield  {title} {\bibinfo {title} {Trading causal order for locality},\ }\href {https://doi.org/10.1103/PhysRevLett.131.120201} {\bibfield  {journal} {\bibinfo  {journal} {Phys. Rev. Lett.}\ }\textbf {\bibinfo {volume} {131}},\ \bibinfo {pages} {120201} (\bibinfo {year} {2023})}\BibitemShut {NoStop}%
\bibitem [{\citenamefont {Dourdent}\ \emph {et~al.}(2025)\citenamefont {Dourdent}, \citenamefont {Simonov}, \citenamefont {Leitherer}, \citenamefont {Boghiu}, \citenamefont {Kunjwal}, \citenamefont {Halder}, \citenamefont {Augusiak},\ and\ \citenamefont {Ac{\'\i}n}}]{DSL25}%
  \BibitemOpen
  \bibfield  {author} {\bibinfo {author} {\bibfnamefont {H.}~\bibnamefont {Dourdent}}, \bibinfo {author} {\bibfnamefont {K.}~\bibnamefont {Simonov}}, \bibinfo {author} {\bibfnamefont {A.}~\bibnamefont {Leitherer}}, \bibinfo {author} {\bibfnamefont {E.-C.}\ \bibnamefont {Boghiu}}, \bibinfo {author} {\bibfnamefont {R.}~\bibnamefont {Kunjwal}}, \bibinfo {author} {\bibfnamefont {S.}~\bibnamefont {Halder}}, \bibinfo {author} {\bibfnamefont {R.}~\bibnamefont {Augusiak}},\ and\ \bibinfo {author} {\bibfnamefont {A.}~\bibnamefont {Ac{\'\i}n}},\ }\bibfield  {title} {\bibinfo {title} {Paradox-free classical non-causality and unambiguous non-locality without entanglement are equivalent},\ }\href@noop {} {\bibfield  {journal} {\bibinfo  {journal} {arXiv preprint arXiv:2512.23599}\ } (\bibinfo {year} {2025})}\BibitemShut {NoStop}%
\bibitem [{\citenamefont {Groisman}\ and\ \citenamefont {Vaidman}(2001)}]{GV01}%
  \BibitemOpen
  \bibfield  {author} {\bibinfo {author} {\bibfnamefont {B.}~\bibnamefont {Groisman}}\ and\ \bibinfo {author} {\bibfnamefont {L.}~\bibnamefont {Vaidman}},\ }\bibfield  {title} {\bibinfo {title} {Nonlocal variables with product-state eigenstates},\ }\href {https://doi.org/10.1088/0305-4470/34/35/313} {\bibfield  {journal} {\bibinfo  {journal} {Journal of Physics A: Mathematical and General}\ }\textbf {\bibinfo {volume} {34}},\ \bibinfo {pages} {6881–6889} (\bibinfo {year} {2001})}\BibitemShut {NoStop}%
\bibitem [{\citenamefont {Baumeler}\ and\ \citenamefont {Tselentis}(2021)}]{BT21}%
  \BibitemOpen
  \bibfield  {author} {\bibinfo {author} {\bibfnamefont {{\"{A}}.}~\bibnamefont {Baumeler}}\ and\ \bibinfo {author} {\bibfnamefont {E.}~\bibnamefont {Tselentis}},\ }\bibfield  {title} {\bibinfo {title} {{Equivalence of Grandfather and Information Antinomy Under Intervention}},\ }\href {https://doi.org/10.4204/EPTCS.340.1} {\bibfield  {journal} {\bibinfo  {journal} {Electronic Proceedings in Theoretical Computer Science}\ }\textbf {\bibinfo {volume} {340}},\ \bibinfo {pages} {1} (\bibinfo {year} {2021})}\BibitemShut {NoStop}%
\bibitem [{\citenamefont {Niset}\ and\ \citenamefont {Cerf}(2006)}]{NC06}%
  \BibitemOpen
  \bibfield  {author} {\bibinfo {author} {\bibfnamefont {J.}~\bibnamefont {Niset}}\ and\ \bibinfo {author} {\bibfnamefont {N.~J.}\ \bibnamefont {Cerf}},\ }\bibfield  {title} {\bibinfo {title} {Multipartite nonlocality without entanglement in many dimensions},\ }\href {https://doi.org/10.1103/PhysRevA.74.052103} {\bibfield  {journal} {\bibinfo  {journal} {Phys. Rev. A}\ }\textbf {\bibinfo {volume} {74}},\ \bibinfo {pages} {052103} (\bibinfo {year} {2006})}\BibitemShut {NoStop}%
\bibitem [{\citenamefont {Oreshkov}\ \emph {et~al.}(2012)\citenamefont {Oreshkov}, \citenamefont {Costa},\ and\ \citenamefont {Brukner}}]{OCB12}%
  \BibitemOpen
  \bibfield  {author} {\bibinfo {author} {\bibfnamefont {O.}~\bibnamefont {Oreshkov}}, \bibinfo {author} {\bibfnamefont {F.}~\bibnamefont {Costa}},\ and\ \bibinfo {author} {\bibfnamefont {{\v{C}}.}~\bibnamefont {Brukner}},\ }\bibfield  {title} {\bibinfo {title} {Quantum correlations with no causal order},\ }\href {https://doi.org/10.1038/ncomms2076} {\bibfield  {journal} {\bibinfo  {journal} {Nature communications}\ }\textbf {\bibinfo {volume} {3}},\ \bibinfo {pages} {1092} (\bibinfo {year} {2012})}\BibitemShut {NoStop}%
\bibitem [{\citenamefont {Ara{\'{u}}jo}\ \emph {et~al.}(2017)\citenamefont {Ara{\'{u}}jo}, \citenamefont {Feix}, \citenamefont {Navascu{\'{e}}s},\ and\ \citenamefont {Brukner}}]{AFN17}%
  \BibitemOpen
  \bibfield  {author} {\bibinfo {author} {\bibfnamefont {M.}~\bibnamefont {Ara{\'{u}}jo}}, \bibinfo {author} {\bibfnamefont {A.}~\bibnamefont {Feix}}, \bibinfo {author} {\bibfnamefont {M.}~\bibnamefont {Navascu{\'{e}}s}},\ and\ \bibinfo {author} {\bibfnamefont {{\v{C}}.}~\bibnamefont {Brukner}},\ }\bibfield  {title} {\bibinfo {title} {A purification postulate for quantum mechanics with indefinite causal order},\ }\href {https://doi.org/10.22331/q-2017-04-26-10} {\bibfield  {journal} {\bibinfo  {journal} {{Quantum}}\ }\textbf {\bibinfo {volume} {1}},\ \bibinfo {pages} {10} (\bibinfo {year} {2017})}\BibitemShut {NoStop}%
\bibitem [{\citenamefont {Nielsen}\ and\ \citenamefont {Chuang}(2010)}]{NC10}%
  \BibitemOpen
  \bibfield  {author} {\bibinfo {author} {\bibfnamefont {M.~A.}\ \bibnamefont {Nielsen}}\ and\ \bibinfo {author} {\bibfnamefont {I.~L.}\ \bibnamefont {Chuang}},\ }\href@noop {} {\emph {\bibinfo {title} {Quantum Computation and Quantum Information: 10th Anniversary Edition}}}\ (\bibinfo  {publisher} {Cambridge University Press},\ \bibinfo {year} {2010})\BibitemShut {NoStop}%
\bibitem [{\citenamefont {Barenco}\ \emph {et~al.}(1995)\citenamefont {Barenco}, \citenamefont {Bennett}, \citenamefont {Cleve}, \citenamefont {DiVincenzo}, \citenamefont {Margolus}, \citenamefont {Shor}, \citenamefont {Sleator}, \citenamefont {Smolin},\ and\ \citenamefont {Weinfurter}}]{BBC95}%
  \BibitemOpen
  \bibfield  {author} {\bibinfo {author} {\bibfnamefont {A.}~\bibnamefont {Barenco}}, \bibinfo {author} {\bibfnamefont {C.~H.}\ \bibnamefont {Bennett}}, \bibinfo {author} {\bibfnamefont {R.}~\bibnamefont {Cleve}}, \bibinfo {author} {\bibfnamefont {D.~P.}\ \bibnamefont {DiVincenzo}}, \bibinfo {author} {\bibfnamefont {N.}~\bibnamefont {Margolus}}, \bibinfo {author} {\bibfnamefont {P.}~\bibnamefont {Shor}}, \bibinfo {author} {\bibfnamefont {T.}~\bibnamefont {Sleator}}, \bibinfo {author} {\bibfnamefont {J.~A.}\ \bibnamefont {Smolin}},\ and\ \bibinfo {author} {\bibfnamefont {H.}~\bibnamefont {Weinfurter}},\ }\bibfield  {title} {\bibinfo {title} {Elementary gates for quantum computation},\ }\href {https://doi.org/10.1103/physreva.52.3457} {\bibfield  {journal} {\bibinfo  {journal} {Physical Review A}\ }\textbf {\bibinfo {volume} {52}},\ \bibinfo {pages} {3457–3467} (\bibinfo {year} {1995})}\BibitemShut {NoStop}%
\bibitem [{\citenamefont {Baumeler}\ and\ \citenamefont {Wolf}(2016{\natexlab{b}})}]{BW16fp}%
  \BibitemOpen
  \bibfield  {author} {\bibinfo {author} {\bibfnamefont {{\"{A}}.}~\bibnamefont {Baumeler}}\ and\ \bibinfo {author} {\bibfnamefont {S.}~\bibnamefont {Wolf}},\ }\bibfield  {title} {\bibinfo {title} {{Device-independent test of causal order and relations to fixed-points}},\ }\href {https://doi.org/10.1088/1367-2630/18/3/035014} {\bibfield  {journal} {\bibinfo  {journal} {New Journal of Physics}\ }\textbf {\bibinfo {volume} {18}},\ \bibinfo {pages} {035014} (\bibinfo {year} {2016}{\natexlab{b}})}\BibitemShut {NoStop}%
\bibitem [{\citenamefont {DiVincenzo}\ \emph {et~al.}(2002)\citenamefont {DiVincenzo}, \citenamefont {Leung},\ and\ \citenamefont {Terhal}}]{DPL02}%
  \BibitemOpen
  \bibfield  {author} {\bibinfo {author} {\bibfnamefont {D.}~\bibnamefont {DiVincenzo}}, \bibinfo {author} {\bibfnamefont {D.}~\bibnamefont {Leung}},\ and\ \bibinfo {author} {\bibfnamefont {B.}~\bibnamefont {Terhal}},\ }\bibfield  {title} {\bibinfo {title} {Quantum data hiding},\ }\href {https://doi.org/10.1109/18.985948} {\bibfield  {journal} {\bibinfo  {journal} {IEEE Transactions on Information Theory}\ }\textbf {\bibinfo {volume} {48}},\ \bibinfo {pages} {580–598} (\bibinfo {year} {2002})}\BibitemShut {NoStop}%
\bibitem [{\citenamefont {Halder}\ \emph {et~al.}(2019)\citenamefont {Halder}, \citenamefont {Banik}, \citenamefont {Agrawal},\ and\ \citenamefont {Bandyopadhyay}}]{HBA19}%
  \BibitemOpen
  \bibfield  {author} {\bibinfo {author} {\bibfnamefont {S.}~\bibnamefont {Halder}}, \bibinfo {author} {\bibfnamefont {M.}~\bibnamefont {Banik}}, \bibinfo {author} {\bibfnamefont {S.}~\bibnamefont {Agrawal}},\ and\ \bibinfo {author} {\bibfnamefont {S.}~\bibnamefont {Bandyopadhyay}},\ }\bibfield  {title} {\bibinfo {title} {Strong quantum nonlocality without entanglement},\ }\bibfield  {journal} {\bibinfo  {journal} {Physical Review Letters}\ }\textbf {\bibinfo {volume} {122}},\ \href {https://doi.org/10.1103/physrevlett.122.040403} {10.1103/physrevlett.122.040403} (\bibinfo {year} {2019})\BibitemShut {NoStop}%
\bibitem [{\citenamefont {Bandyopadhyay}(2010)}]{B10}%
  \BibitemOpen
  \bibfield  {author} {\bibinfo {author} {\bibfnamefont {S.}~\bibnamefont {Bandyopadhyay}},\ }\bibfield  {title} {\bibinfo {title} {Entanglement and perfect discrimination of a class of multiqubit states by local operations and classical communication},\ }\href {https://doi.org/10.1103/PhysRevA.81.022327} {\bibfield  {journal} {\bibinfo  {journal} {Phys. Rev. A}\ }\textbf {\bibinfo {volume} {81}},\ \bibinfo {pages} {022327} (\bibinfo {year} {2010})}\BibitemShut {NoStop}%
\bibitem [{\citenamefont {Bandyopadhyay}(2011)}]{B11}%
  \BibitemOpen
  \bibfield  {author} {\bibinfo {author} {\bibfnamefont {S.}~\bibnamefont {Bandyopadhyay}},\ }\bibfield  {title} {\bibinfo {title} {More nonlocality with less purity},\ }\bibfield  {journal} {\bibinfo  {journal} {Physical Review Letters}\ }\textbf {\bibinfo {volume} {106}},\ \href {https://doi.org/10.1103/physrevlett.106.210402} {10.1103/physrevlett.106.210402} (\bibinfo {year} {2011})\BibitemShut {NoStop}%
\bibitem [{\citenamefont {Halder}\ and\ \citenamefont {Augusiak}(2025)}]{HR25}%
  \BibitemOpen
  \bibfield  {author} {\bibinfo {author} {\bibfnamefont {S.}~\bibnamefont {Halder}}\ and\ \bibinfo {author} {\bibfnamefont {R.}~\bibnamefont {Augusiak}},\ }\href {https://arxiv.org/abs/2503.05673} {\bibinfo {title} {Local unambiguous unidentifiability, entanglement generation, and hilbert space splitting}} (\bibinfo {year} {2025}),\ \Eprint {https://arxiv.org/abs/2503.05673} {arXiv:2503.05673 [quant-ph]} \BibitemShut {NoStop}%
\bibitem [{\citenamefont {Bennett}\ \emph {et~al.}(1999{\natexlab{b}})\citenamefont {Bennett}, \citenamefont {DiVincenzo}, \citenamefont {Mor}, \citenamefont {Shor}, \citenamefont {Smolin},\ and\ \citenamefont {Terhal}}]{BDP99UPB}%
  \BibitemOpen
  \bibfield  {author} {\bibinfo {author} {\bibfnamefont {C.~H.}\ \bibnamefont {Bennett}}, \bibinfo {author} {\bibfnamefont {D.~P.}\ \bibnamefont {DiVincenzo}}, \bibinfo {author} {\bibfnamefont {T.}~\bibnamefont {Mor}}, \bibinfo {author} {\bibfnamefont {P.~W.}\ \bibnamefont {Shor}}, \bibinfo {author} {\bibfnamefont {J.~A.}\ \bibnamefont {Smolin}},\ and\ \bibinfo {author} {\bibfnamefont {B.~M.}\ \bibnamefont {Terhal}},\ }\bibfield  {title} {\bibinfo {title} {Unextendible product bases and bound entanglement},\ }\href {https://doi.org/10.1103/physrevlett.82.5385} {\bibfield  {journal} {\bibinfo  {journal} {Physical Review Letters}\ }\textbf {\bibinfo {volume} {82}},\ \bibinfo {pages} {5385–5388} (\bibinfo {year} {1999}{\natexlab{b}})}\BibitemShut {NoStop}%
\bibitem [{\citenamefont {DiVincenzo}\ \emph {et~al.}(2003)\citenamefont {DiVincenzo}, \citenamefont {Mor}, \citenamefont {Shor}, \citenamefont {Smolin},\ and\ \citenamefont {Terhal}}]{DPM03UPB}%
  \BibitemOpen
  \bibfield  {author} {\bibinfo {author} {\bibfnamefont {D.~P.}\ \bibnamefont {DiVincenzo}}, \bibinfo {author} {\bibfnamefont {T.}~\bibnamefont {Mor}}, \bibinfo {author} {\bibfnamefont {P.~W.}\ \bibnamefont {Shor}}, \bibinfo {author} {\bibfnamefont {J.~A.}\ \bibnamefont {Smolin}},\ and\ \bibinfo {author} {\bibfnamefont {B.~M.}\ \bibnamefont {Terhal}},\ }\bibfield  {title} {\bibinfo {title} {Unextendible product bases, uncompletable product bases and bound entanglement},\ }\href {https://doi.org/10.1007/s00220-003-0877-6} {\bibfield  {journal} {\bibinfo  {journal} {Communications in Mathematical Physics}\ }\textbf {\bibinfo {volume} {238}},\ \bibinfo {pages} {379–410} (\bibinfo {year} {2003})}\BibitemShut {NoStop}%
\bibitem [{\citenamefont {Ara\'ujo}\ \emph {et~al.}(2017)\citenamefont {Ara\'ujo}, \citenamefont {Gu\'erin},\ and\ \citenamefont {Baumeler}}]{AGB17}%
  \BibitemOpen
  \bibfield  {author} {\bibinfo {author} {\bibfnamefont {M.}~\bibnamefont {Ara\'ujo}}, \bibinfo {author} {\bibfnamefont {P.~A.}\ \bibnamefont {Gu\'erin}},\ and\ \bibinfo {author} {\bibfnamefont {A.}~\bibnamefont {Baumeler}},\ }\bibfield  {title} {\bibinfo {title} {Quantum computation with indefinite causal structures},\ }\href {https://doi.org/10.1103/PhysRevA.96.052315} {\bibfield  {journal} {\bibinfo  {journal} {Phys. Rev. A}\ }\textbf {\bibinfo {volume} {96}},\ \bibinfo {pages} {052315} (\bibinfo {year} {2017})}\BibitemShut {NoStop}%
\bibitem [{\citenamefont {Baumeler}\ \emph {et~al.}(2022)\citenamefont {Baumeler}, \citenamefont {Gilani},\ and\ \citenamefont {Rashid}}]{BGR22}%
  \BibitemOpen
  \bibfield  {author} {\bibinfo {author} {\bibfnamefont {A.}~\bibnamefont {Baumeler}}, \bibinfo {author} {\bibfnamefont {A.~S.}\ \bibnamefont {Gilani}},\ and\ \bibinfo {author} {\bibfnamefont {J.}~\bibnamefont {Rashid}},\ }\bibfield  {title} {\bibinfo {title} {Unlimited non-causal correlations and their relation to non-locality},\ }\href {https://doi.org/10.22331/q-2022-03-29-673} {\bibfield  {journal} {\bibinfo  {journal} {Quantum}\ }\textbf {\bibinfo {volume} {6}},\ \bibinfo {pages} {673} (\bibinfo {year} {2022})}\BibitemShut {NoStop}%
\bibitem [{\citenamefont {Erde}(2024)}]{matchings}%
  \BibitemOpen
  \bibfield  {author} {\bibinfo {author} {\bibfnamefont {J.}~\bibnamefont {Erde}},\ }\href {https://arxiv.org/abs/2404.03950} {\bibinfo {title} {Matchings in the hypercube with specified edges}} (\bibinfo {year} {2024}),\ \Eprint {https://arxiv.org/abs/2404.03950} {arXiv:2404.03950 [math.CO]} \BibitemShut {NoStop}%
\bibitem [{\citenamefont {Baumeler}(2017)}]{B17}%
  \BibitemOpen
  \bibfield  {author} {\bibinfo {author} {\bibfnamefont {A.}~\bibnamefont {Baumeler}},\ }\bibfield  {title} {\bibinfo {title} {Causal loops: logically consistent correlations, time travel, and computation},\ }\href {https://cqi.inf.usi.ch/publications/these_amin.pdf} {\bibfield  {journal} {\bibinfo  {journal} {PhD thesis, Università della Svizzera italiana}\ } (\bibinfo {year} {2017})}\BibitemShut {NoStop}%
\bibitem [{\citenamefont {Baumeler}\ \emph {et~al.}(2014)\citenamefont {Baumeler}, \citenamefont {Feix},\ and\ \citenamefont {Wolf}}]{BFW14}%
  \BibitemOpen
  \bibfield  {author} {\bibinfo {author} {\bibfnamefont {A.}~\bibnamefont {Baumeler}}, \bibinfo {author} {\bibfnamefont {A.}~\bibnamefont {Feix}},\ and\ \bibinfo {author} {\bibfnamefont {S.}~\bibnamefont {Wolf}},\ }\bibfield  {title} {\bibinfo {title} {Maximal incompatibility of locally classical behavior and global causal order in multiparty scenarios},\ }\href {https://doi.org/10.1103/PhysRevA.90.042106} {\bibfield  {journal} {\bibinfo  {journal} {Phys. Rev. A}\ }\textbf {\bibinfo {volume} {90}},\ \bibinfo {pages} {042106} (\bibinfo {year} {2014})}\BibitemShut {NoStop}%
\bibitem [{\citenamefont {Akibue}(2022)}]{A22}%
  \BibitemOpen
  \bibfield  {author} {\bibinfo {author} {\bibfnamefont {S.}~\bibnamefont {Akibue}},\ }\href {https://arxiv.org/abs/2202.06518} {\bibinfo {title} {Entanglement and causal relation in distributed quantum computation}} (\bibinfo {year} {2022}),\ \Eprint {https://arxiv.org/abs/2202.06518} {arXiv:2202.06518 [quant-ph]} \BibitemShut {NoStop}%
\bibitem [{\citenamefont {Bravyi}\ and\ \citenamefont {Kitaev}(2005)}]{BK04}%
  \BibitemOpen
  \bibfield  {author} {\bibinfo {author} {\bibfnamefont {S.}~\bibnamefont {Bravyi}}\ and\ \bibinfo {author} {\bibfnamefont {A.}~\bibnamefont {Kitaev}},\ }\bibfield  {title} {\bibinfo {title} {Universal quantum computation with ideal clifford gates and noisy ancillas},\ }\bibfield  {journal} {\bibinfo  {journal} {Physical Review A}\ }\textbf {\bibinfo {volume} {71}},\ \href {https://doi.org/10.1103/physreva.71.022316} {10.1103/physreva.71.022316} (\bibinfo {year} {2005})\BibitemShut {NoStop}%
\bibitem [{\citenamefont {Chiribella}\ \emph {et~al.}(2009)\citenamefont {Chiribella}, \citenamefont {D’Ariano},\ and\ \citenamefont {Perinotti}}]{CDP09}%
  \BibitemOpen
  \bibfield  {author} {\bibinfo {author} {\bibfnamefont {G.}~\bibnamefont {Chiribella}}, \bibinfo {author} {\bibfnamefont {G.~M.}\ \bibnamefont {D’Ariano}},\ and\ \bibinfo {author} {\bibfnamefont {P.}~\bibnamefont {Perinotti}},\ }\bibfield  {title} {\bibinfo {title} {Theoretical framework for quantum networks},\ }\bibfield  {journal} {\bibinfo  {journal} {Physical Review A}\ }\textbf {\bibinfo {volume} {80}},\ \href {https://doi.org/10.1103/physreva.80.022339} {10.1103/physreva.80.022339} (\bibinfo {year} {2009})\BibitemShut {NoStop}%
\end{thebibliography}%
\newpage
\appendix
\onecolumngrid
\section{Proof of Theorem VII.1}\label{sec:appendix-a}
We present two results needed for the proof of Theorem \ref{th:purification}. The first one, proved in Section 4.7.1 of Ref.~\cite{B17}, is that any process function can be turned into a reversible process function. The second, that any process function can be mapped to a unitary operator, is then obtained as corollary. 
\begin{theorem}[Reversible process function, adapted from Theorem 7 in Section 4.7.1 of Ref.~\cite{B17}]\label{th:rev-pf}
Any n-partite process function $\omega : \vec{O}\rightarrow \vec{I}$ can be turned into a reversible function $\omega_r : \vec{O}\times O_{src} \rightarrow \vec{I} \times I_{sink}$ defined as $\omega_r(\vec{o}, \vec{s}) = (\vec{s} \oplus_{\vec{I}} \omega(\vec{o}),\vec{o})$, where $\oplus_{\vec{I}}$ is the addition applied locally modulo $|I_k| ~\forall k \in [n]$ and we have additional systems, a source system with output space $O_{src}  = \vec{I}$ and trivial input space, and a sink system with input space $I_{sink} = \vec{O}$ and trivial output space.
\end{theorem}
It was already proven that the AF/BW process can be mapped into a unitary operator (see Ref.~\cite{AFN17} and Section 4.8.2 of Ref.~\cite{B17}). We generalize this result to any process function as corollary of the previous theorem. 
\begin{corollary}\label{cor:unitary-pf}
Any process function can be turned into a unitary operator.
\end{corollary}
\begin{proof}
By Theorem \ref{th:rev-pf}, any process function $\omega$ can be mapped into a reversible process function $\omega_r : (\vec{o}, \vec{s}) \mapsto (\vec{s} \oplus_{\vec{I}} \omega(\vec{o}),\vec{o})$. It suffices then to define the following operator from the spaces $\vec{O}\times O_{src}$ to the spaces $\vec{I}\times I_{sink}$, encoding the action of $\omega_r$: 
\begin{equation}\label{eq:unitary-pf}
U_\omega = \sum\limits_{\vec{o}, \vec{s}} \ketbra{\vec{s} \oplus_{\vec{I}} \omega(\vec{o})}{\vec{o}} \otimes \ketbra{\vec{o}}{\vec{s}},
\end{equation}
which is unitary because, by bijectivity of the map $(\vec{o}, \vec{s}) \mapsto (\vec{s} \oplus_{\vec{I}} \omega(\vec{o}),\vec{o})$, each state $\ket{\vec{s} \oplus_{\vec{I}} \omega(\vec{o})}\otimes\ket{\vec{o}}$ is obtained from the application of $U_\omega$ to a unique state  $\ket{\vec{o}}\otimes\ket{\vec{s}}$, essentially permuting the states of the $|\vec{O}|\cdot|\vec{I}|$-dimensional computational basis.
\end{proof}
We can now prove Theorem \ref{th:purification}.
\begin{proof}  
We begin to show that the first step of the protocol allows us to extract the process function. We then show that the standard unitary purification of this process function generates a valid unitary process.
\begin{enumerate}
    \item Extraction of the process function:
    Consider the basis $B =\{\bigotimes\limits_{k = 1}^n U^{(k)}_{i_k = \omega_k(\vec{o}_{\backslash k})}\ket{\vec{o}} | \vec{o} \in \vec{O}\}$ obtained from the application of the PFU to the computational basis. By Corollary \ref{cor:local-unitaries}, the labels define valid local functions $f_k$. Additionally, since $B$ is generated by a PFU, the labels correspond to a process function. Indeed, each vector will be assigned a label $(\vec{o}|\vec{i})$ such that $\vec{i} = \omega(\vec{o})$, with $\omega$ being the process function in the PFU decomposition. We now show how the function $f = (f_k)_{k \in [n]}: \vec{L} \rightarrow \vec{I}$ correctly maps to the right $\vec{i} = \omega(\vec{o})$ for every input state $\ket{\vec{o}}$ to the extraction circuit (\textit{e.g.}, in Figure \ref{fig:unitarypurification}). 
    
    If $\ket{\vec{o}}$ is in $B$, the measurement device deterministically outputs the global classical description $\vec{o}$, to which the function $f$ correctly assigns the global basis label $\vec{0} = \omega(\vec{o})$ (remember that $i_k = 0$ corresponds to the identity). If $\ket{\vec{o}}$ is not in $B$, the device acts probabilistically, with the possible measurement outcomes being the states in $B$ that are non-orthogonal to $\ket{\vec{o}}$. Since the control conditions are mutually exclusive, the string $\vec{o}=(o_k,\vec{o}_{\backslash k})$ corresponds to a unique $i_k$ (the one labelling the subset $B_{k,i_k}\subseteq\vec{O}$ in Definition \ref{def:MECC-U}) for some $k$. Therefore, the measurement outcomes non-orthogonal to $\ket{\vec{o}}$ are necessarily the states in the subset \{$(U^{(k)}_{i_k} \otimes \id_{\backslash k})\ket{\vec{o}}|\vec{o}\in B_{k,i_k}$\}, which will be mapped by $f$ to the same global basis label $i_k\vec{e}_k = \omega(\vec{o})$.

    \item Purification of the process function:
    
By Theorem \ref{th:rev-pf}, the extracted process function can be transformed into its reversible version $\omega_r : (\vec{o}, \vec{s}) \mapsto (\vec{s} \oplus_{\vec{I}} \omega(\vec{o}),\vec{o})$. It remains to show that the Choi–Jamiołkowski (CJ) operator $W = \dketbra{U_\omega}{U_\omega}$ of the unitary $U_\omega$ defined in Eq~\ref{eq:unitary-pf} is a valid process, where 
\begin{equation}
 \dket{U_\omega} = \sum\limits_{\vec{o},\vec{s}} \ket{\vec{o}}_{\vec{O}}\otimes\ket{\vec{s}}_{O_{src}}\otimes\ket{\vec{s}\oplus_{\vec{I}} \omega(\vec{o})}_{\vec{I}}\otimes\ket{\vec{o}}_{I_{sink}}.   
\end{equation}
A process is valid if it maps any local CPTP map $\mathcal{M} = \bigotimes\limits_{k=1}^n \mathcal{M}^k$ from $\vec{I}$ to $\vec{O}$ to valid a CPTP map $\mathcal{N}$ from $O_{src}$ to $I_{sink}$. Reformulating the validity of $W$ in terms of $N$, the CJ operator of the map $\mathcal{N}$, we have that: 
\begin{align}
 N &= \Tr_{\vec{I},\vec{O}}[W(\bigotimes_{k=1}^n {M^k}^T\otimes \id_{O_{src}I_{sink}})] \nonumber  \\ 
 &= M * W \ \text{s.t.} \\ 
 N  &\geq 0, \Tr_{I_{sink}}[N] =  \id_{O_{src}}, \label{eq:valid}
\end{align}
where * is the link product as defined in \cite{CDP09}) and $M^k$ are the CJ operators of the parties's local interventions $\mathcal{M}^k$. 

In Ref.~\cite{AGB17}, M. Araújo et al. define an $n$-partite generalization of the AF/BW process (with binary inputs and outputs) that corresponds to MECC unitaries (e.g., the 4-party example in Table \ref{tab:MECC-PF}) and prove the validity of its corresponding $W = \dketbra{U_\omega}{U_\omega}$ in Appendix B. Since they make explicit use of the MECC property, the rest of the proof essentially follows the steps in Appendix B of \cite{AGB17}, generalizing them to MECC process function with arbitrary inputs and outputs. 

First, we express the local map $M$ in terms of its Kraus operators $M_l$:
\begin{equation}
M = \sum\limits_{\vec{i},\vec{j},l}\ketbra{\vec{i}}{\vec{j}}_{\vec{I}} \otimes M_l\ketbra{\vec{i}}{\vec{j}}_{\vec{I}}M_l^\dag.
\end{equation}
We then rewrite the process vector $\dket{U_\omega}$ as
\begin{equation}
\dket{U_\omega} = \sum\limits_{\vec{o},\vec{s}} \ket{\vec{o}}_{\vec{O}}\otimes\ket{\vec{s}}_{O_{src}}\otimes\Pi_{\omega(\vec{o})}\ket{\vec{s}}_{\vec{I}}\otimes\ket{\vec{o}}_{I_{sink}}, 
\end{equation}
where $\Pi_{\omega(\vec{o})} = \bigotimes\limits_{k=1}^n \Pi_{\omega_k(\vec{o}_{\backslash k})}$, with $\Pi_{\omega_k(\vec{o}_{\backslash k})}$ being a local permutation mapping $\ket{s_k}_{I_k}$ to $\ket{s_k\oplus_{|I_k|} \omega_k(\vec{o}_{\backslash k})}_{I_k}$.

After a few calculations, we obtain the operator
\begin{equation}
    N = \sum\limits_{\vec{o},\vec{\Tilde{o}}, l}\Pi^T_{\omega(\vec{o})}{M_l}^T\ketbra{\vec{o}}{\vec{\Tilde{o}}}(M_l^\dag)^T\Pi_{\omega(\vec{\Tilde{o}})} \otimes \ketbra{\vec{o}}{\vec{\Tilde{o}}}_{I_{sink}},
\end{equation}
which is positive since the link product of two positive operators is also positive \cite{CDP09}. Imposing the trace preserving condition (\ref{eq:valid}) and transposing both sides gives
\begin{equation}
\sum\limits_{\vec{o},l}\Pi^T_{\omega(\vec{o})}{M_l}^\dag\ketbra{\vec{o}}{\vec{o}}M_l\Pi_{\omega(\vec{o})} = \id_{O_{src}}.
\end{equation}
Because of the MECC property, at most one of $\omega$'s component takes a non-zero value while the remaining are all null. We can thus decompose the left-hand-side as
\begin{align}
&\sum\limits_{l, \vec{o}:\omega(\vec{o})=\vec{0}}{M_l}^\dag\ketbra{\vec{o}}{\vec{o}}M_l \label{eq:valid-mecc1} \\  
&+ \sum\limits_{k = 1}^n \sum\limits_{i_k = 1}^{|I_k|} \sum\limits_{l, \vec{o}:\omega_k(\vec{o}_{\backslash k})=i_k}\Pi^T_{\omega(\vec{o})}{M_l}^\dag\ketbra{\vec{o}}{\vec{o}}M_l\Pi_{\omega(\vec{o})}, \label{eq:valid-mecc2}
\end{align}
Let's look at one term in line (\ref{eq:valid-mecc2}) for some $k'$ and $i_{k'} \neq 0$:
\begin{align}
&\sum\limits_{l, \vec{o}:\omega_{k'}(\vec{o}_{\backslash k'})=i_{k'}}\Pi^T_{\omega(\vec{o})}{M_l}^\dag\ketbra{\vec{o}}{\vec{o}}M_l\Pi_{\omega(\vec{o})}  \\
&= \sum_l (\Pi^T_{i_{k'}} \otimes \id_{\backslash k'})M^\dag_l(\sum\limits_{\vec{o}:\omega_{k'}(\vec{o}_{\backslash k'})=i_{k'}} \ketbra{\vec{o}} {\vec{o}})M_l(\Pi_{i_{k'}} \otimes \id_{\backslash k'}) \\
&= \sum_l (\Pi^T_{i_{k'}} \otimes \id_{\backslash {k'}})M^\dag_l(\id_{k'} \otimes \ketbra{\vec{c}_{k',i_{k'}}}{\vec{c}_{k',i_{k'}}})M_l(\Pi_{i_{k'}} \otimes \id_{\backslash k'})\\
&=\sum\limits_{l_k: k\neq k'}\Motimes_{k\neq k'}{M^k_{l_k}}^\dag\bigg(\Pi^T_{i_{k'}}(\sum_{l_{k'}}M^{k' \dag}_{l_{k'}}M^{k'}_{l_{k'}})\Pi_{i_{k'}} \otimes \ketbra{\vec{c}_{k',i_{k'}}}{\vec{c}_{k',i_{k'}}}\bigg)\Motimes_{k\neq k'}M^k_{l_k} \\
&=\sum\limits_{l_k: k\neq k'}\Motimes_{k\neq k'}{M^k_{l_k}}^\dag\bigg(\id_{k'} \otimes \ketbra{\vec{c}_{k',i_{k'}}}{\vec{c}_{k',i_{k'}}}\bigg)\Motimes_{k\neq k'}M^k_{l_k} \\
&= \sum_l M^\dag_l(\sum\limits_{\vec{o}:\omega_{k'}(\vec{o}_{\backslash k'})=i_{k'}} \ketbra{\vec{o}} {\vec{o}})M_l \\
&=\sum\limits_{l, \vec{o}:\omega_{k'}(\vec{o}_{\backslash k'})=i_{k'}}{M_l}^\dag\ketbra{\vec{o}}{\vec{o}}M_l.
\end{align}

This being true for every party $k$ and every input $i_k$, the operator $\Pi_{\omega(\vec{o})}$ in line (\ref{eq:valid-mecc2}) cancels itself out:
\begin{align}
&\sum\limits_{k = 1}^n \sum\limits_{i_k = 1}^{|I_k|} \sum\limits_{l, \vec{o}:\omega_k(\vec{o}_{\backslash k})=i_k}\Pi^T_{\omega(\vec{o})}{M_l}^\dag\ketbra{\vec{o}}{\vec{o}}M_l\Pi_{\omega(\vec{o})} \\
&= \sum\limits_{k = 1}^n \sum\limits_{i_k = 1}^{|I_k|} \sum\limits_{l, \vec{o}:\omega_k(\vec{o}_{\backslash k})=i_k}{M_l}^\dag\ketbra{\vec{o}}{\vec{o}}M_l.\label{eq:valid-mecc3}
\end{align}
Finally, combining (\ref{eq:valid-mecc3}) with line (\ref{eq:valid-mecc1}), we have that
\begin{equation}
\sum\limits_{\vec{o},l}\Pi^T_{\omega(\vec{o})}{M_l}^\dag\ketbra{\vec{o}}{\vec{o}}M_l\Pi_{\omega(\vec{o})} = \sum\limits_{l, \vec{o}}{M_l}^\dag\ketbra{\vec{o}}{\vec{o}}M_l = \id_{O'}.
\end{equation}
\end{enumerate}
\end{proof}
\end{document}